\documentclass[11pt,letterpaper,oneside,english]{article}
\usepackage{fullpage} 

\usepackage{mdwlist}

\usepackage{amssymb,amsbsy,latexsym}
\usepackage{amsmath}
\usepackage{graphics, subfigure, float}

\usepackage[latin1]{inputenc}
\usepackage[american]{babel}
\usepackage[T1]{fontenc} 

\makeatletter
\let\@font@warningori\@font@warning
\newcommand\shutup{\def\@font@warning##1{}}
\newcommand\youcanspeak{\let\@font@warning\@font@warningori}
\makeatother 
\usepackage{amsmath,amsbsy,bm,latexsym}
\shutup
\usepackage{fourier}
\youcanspeak
\usepackage[scaled=0.875]{helvet}

\usepackage{amscd,amsthm}
\usepackage{hyperref}

\usepackage[dvips]{graphicx,epsfig,color}
\usepackage{pst-all}
\usepackage{pstricks-add}
\newpsobject{showgrid}{psgrid}{subgriddiv=1,griddots=10,gridlabels=6pt}
\usepackage{bm}

\newtheoremstyle{theorem}{1em}{1em}{\slshape}{0pt}{\bfseries}{.}{ }{}
\theoremstyle{theorem}
\newtheorem{theorem}{Theorem}
\newtheorem*{theorem*}{Theorem}
\newtheorem{corollary}[theorem]{Corollary}

\newtheorem{lemma}[theorem]{Lemma}
\newtheorem*{lemma*}{Lemma}

\newtheorem*{claim*}{Claim}
\newtheorem*{claimI}{Claim (I)}
\newtheorem*{claimII}{Claim (II)}
\newtheorem{fact}[theorem]{Fact}

\newtheorem*{definition*}{Definition}

\theoremstyle{remark}

\newtheorem*{remark*}{Remark}

\newenvironment{proofofclaim}{\vspace{1ex}\noindent{\emph{Proof of claim.}}\hspace{0.5em}}
   	    {\hfill$\lozenge$\vspace{1ex}}

\providecommand{\setN}{\mathbb{N}}
\providecommand{\setZ}{\mathbb{Z}}
\providecommand{\setQ}{\mathbb{Q}}
\providecommand{\setR}{\mathbb{R}}

\newcommand{\disc}{\textrm{disc}}

\newcommand{\setS}{\bm{S}}
\newcommand{\setC}{\bm{C}} 
\newcommand{\setY}{\bm{Y}}
\newcommand{\setD}{\bm{D}}
\newcommand{\supp}{\textrm{supp}}

\floatstyle{ruled}
\newfloat{algorithm}{tbp}{loa}
\floatname{algorithm}{Algorithm}

\usepackage[displaymath,textmath,graphics, subfigure]{preview} 
\PreviewEnvironment{myproblem}
\PreviewEnvironment{algorithm} 
\PreviewEnvironment{defn} 
\PreviewEnvironment{center} 
\PreviewEnvironment{pspicture} 

\title{The Entropy Rounding Method in Approximation Algorithms}
\author{{\Large{Thomas Rothvoß}}\thanks{Supported by the Alexander von Humboldt Foundation within the Feodor Lynen program.}\vspace{3mm}\\ M.I.T. \\{ \tt{rothvoss@math.mit.edu}}}

\begin{document}

\maketitle
\thispagestyle{empty}

\begin{abstract}
\noindent Let $A$ be a matrix, $c$ be any linear objective function and $x$ be a fractional 
vector, say an LP solution to some discrete optimization problem. Then a recurring
task in theoretical computer science (and in approximation algorithms in particular) 
is to obtain an integral vector $y$ such that
$Ax \approx Ay$ and $c^Ty$ exceeds $c^Tx$ by only a moderate factor.

We give a new randomized rounding procedure for this task, provided that $A$ has 
\emph{bounded $\Delta$-approximate entropy}.
This property means that for uniformly chosen random signs $\chi(j) \in \{ ± 1\}$ 
on any subset of the columns, the outcome $A\chi$ can be approximately described 
using a sub-linear number of bits in expectation. 

To achieve this result, we modify well-known techniques from the field of \emph{discrepancy theory}, 
especially we rely on \emph{Beck's entropy method}, which to the best of our knowledge has never been
used before in the context of approximation algorithms. 
Our result can be made  
constructive using the Bansal framework based on semidefinite programming.

We demonstrate the versatility of our procedure by rounding fractional solutions
to column-based linear programs for some generalizations of {\sc Bin Packing}. 
For example we obtain a polynomial time $OPT + O(\log^2 OPT)$ approximation for 
{\sc Bin Packing With Rejection}
and the first AFPTAS for the {\sc Train Delivery} problem.
\end{abstract}

\newpage
\setcounter{page}{1}

\section{Introduction}

Many approximation algorithms are based on linear programming relaxations; 
for the sake of concreteness, say on formulations like
\[
  \min\left\{ c^Tx \mid Ax \geq b, x \geq \mathbf{0} \right\},
\]
with $A \in \setR^{n × m}$. 
Several techniques have been developed to round
a fractional LP solution $x$
to an integer one; the textbooks \cite{ApproximationAlgorithmsVazirani01, DesignApproximationAlgorithmsShmoysWilliamson11} provide
a good overview on the most common approaches. The aim of this paper is to introduce a new 
LP rounding technique that we term \emph{entropy rounding}. 
 
To describe our method, we consider the random variable $A\chi$, where $\chi \in \{ ± 1\}^m$ is 
a uniformly chosen \emph{random coloring} of the columns of $A$. Suppose that $A$ has the property that one
can approximately encode the outcome of $A\chi$ up to an additive error of
 $\Delta$ with at most $\frac{m}{5}$ bits in expectation. In other words, we suppose
that we can find some arbitrary function $f$ such that $\|A\chi - f(\chi)\|_{\infty} \leq \Delta$
and the \emph{entropy} of the random variables $f(\chi)$ can be bounded by $\frac{m}{5}$. 
Note that the entropy could never exceed $m$, hence we only need to save
a constant factor by allowing an approximation error. 
One possible choice could be
$f(\chi) = 2\Delta \lceil \frac{A\chi}{2\Delta} \rfloor$, meaning that we round every entry of $A\chi$ to the
nearest multiple of $2\Delta$. 
To bound the entropy of $f(\chi)$ one can then use standard concentration bounds 
since the values $A_i\chi = \sum_{j=1}^m A_{ij}\chi(j)$ are the sum of 
independently distributed random variables (here $A_i$ denotes the $i$th row of $A$).
If this holds also for any submatrix of $A$, we say that $A$ has 
\emph{bounded $\Delta$-approximate entropy}. 

But why would it be useful to have this property for $A$? 
Since there are $2^m$
many colorings $\chi$, there must be an exponential number of colorings
$\chi^{(1)},\ldots,\chi^{(\ell)} \in \{ ± 1 \}^m$, which are \emph{similar} w.r.t. $A$, i.e.
$\|A\chi^{(i)} - A\chi^{(i')}\|_{\infty} \leq \Delta$. Since there are so many similar colorings, we can pick
two of them (say $\chi^{(i)},\chi^{(i')}$) that differ in at least half of the entries 
and define $\chi := \frac{1}{2}(\chi^{(i)} - \chi^{(i')})$ as the difference of those
colorings. Then $\chi$ is a \emph{half-coloring}, i.e. it has entries in $\{-1,0,1\}$, 
but at least half of the entries are non-zero and furthermore $\|A\chi \|_{\infty} \leq \Delta$.

However, our aim was to find a vector $y \in \{ 0,1\}^m$ such that $Ay \approx Ax$.
We will iteratively obtain half-colorings and use them to update $x$, each time reducing
its fractionality.  
Thus, we consider the least value bit in any entry of $x$; say this is
bit $K$. Let $J \subseteq [m]$ be the set of indices where this bit is set to one 
and let $A^J \subseteq A$ be the submatrix of the corresponding 
columns. Then by the argument above, there is a half-coloring $\chi \in \{ 0, ± 1\}^J$
such that $\|A^J\chi\|_{\infty} \leq \Delta$. We use this information to round our fractional solution 
to
$x' := x + (\frac{1}{2})^K\chi$, meaning that we delete the $K$th bit
of those entries $j$ that have $\chi(j) = -1$; we round the entry up if $\chi(j) = 1$ and we
leave it unchanged if $\chi(j) = 0$. After iterating this at most $\log m$ times, 
the $K$th bit of all entries of $x$ will be $0$. Hence after 
at most $K\cdot \log m$ iterations, we will end up in a $0/1$ vector that we term $y$. 
This vector satisfies 
$\|Ax - Ay\|_{\infty} \leq \sum_{k=1}^K (\frac{1}{2})^k \cdot \log m \cdot \Delta \leq \log m \cdot \Delta$.

Let us illustrate this abstract situation with a concrete example. For the 
very classical 
 {\sc Bin Packing} problem, the input consists of a sorted list of \emph{item sizes} $1 \geq s_1 \geq \ldots \geq s_n > 0$ 
and the goal is to assign all items to a minimum number of bins of size $1$. 
Let $\setS = \{ S \subseteq [n] \mid \sum_{i\in S} s_i \leq 1\}$ be the set system containing all feasible \emph{patterns} and 
let $\mathbf{1}_S$ denote the characteristic vector of a set $S$. 
A well-studied column-based LP relaxation for {\sc Bin Packing} is
 \begin{equation} \label{eq:LP}
    \min\Big\{ \mathbf{1}^Tx \mid \sum_{S\in\setS} x_S \mathbf{1}_{S} = \mathbf{1}, x \geq \mathbf{0} \Big\}
 \end{equation}
(see e.g. \cite{TrimProblem-Eiseman1957,Gilmore-Gomory61,KarmarkarKarp82}).
In an integral solution, the variable $x_S$ tells whether a bin should be packed
exactly with the items in $S$.
We want to argue why our method is applicable here. Thus let $x$ be a 
fractional solution to \eqref{eq:LP}. In order to keep the
notation simple let us assume for now, that all items have size between $\frac{1}{2k}$
and $\frac{1}{k}$. 
Our choice for matrix $A$ is as follows: Let $A_i$ be the sum of 
the first $i$ rows of the constraint matrix of \eqref{eq:LP}, i.e. $A_{iS} = |S \cap \{ 1,\ldots,i\}|$. 
By definition, for an integral vector $y$,  $A_iy$ denotes the number of slots that $y$ reserves for 
items in $1,\ldots,i$. If there are less than $i$ many slots reserved, we 
term this a \emph{deficit}. Since we assumed that the items are sorted 
according to their size, a vector $y \in \{0,1\}^{\setS}$ will correspond to a feasible solution
if there is no deficit for any interval $1,\ldots,i$. 

To understand why this matrix $A$ has the needed property,
we can add some artificial rows until consecutive rows differ in 
exactly one entry; say $n' \leq mk$ is the new number of rows.
Then observe that the sequence $A_1\chi,A_2\chi,\ldots,A_{n'}\chi$
describes a symmetric random walk with step size 1 on the real axis.
We imagine all multiples of $\Delta$ as ``mile stones'' and 
choose $f_i(\chi)$ as the last such mile stone that was crossed by 
the first $i$ steps of the random walk (i.e. by $A_1\chi,\ldots,A_i\chi$). 
For an independent random walk it would take $\Theta(\Delta^2)$ iterations
in expectation until a random walk covers a distance of $\Delta$, thus
we expect that the sequence $f_1(\chi),\ldots,f_{n'}(\chi)$ changes its value only
every $\Theta(\Delta^2)$ steps and consequently the entropy of this sequence cannot be large.
But up to $k$ steps of the random walk correspond to the same column of $A$ and depend on each other.
Using more involved arguments, we will still be able to show that
for $\Delta := \Theta(\frac{1}{k})$, the entropy of the sequence
  $f_1(\chi),\ldots,f_{n'}(\chi)$ is bounded by $\frac{m}{5}$.

More generally, we allow that the parameter $\Delta$ depends on the row $i$ of $A$.
Then the same arguments go through for $\Delta_i := \Theta(\frac{1}{s_i})$, where $s_i$ is the size of item $i$. 
Thus our rounding procedure can be applied to a fractional {\sc Bin Packing} solution $x$ 
to provide an integral vector $y$ with  $|A_ix - A_iy| \leq O(\log n)\cdot \Delta_i$. The deficits
can be eliminated by buying $O(\log^2 n)$ extra bins in total. 

The entropy-based argument which guarantees the existence of proper
half-colorings $\chi$ is widely termed
``Beck's Entropy Method'' from the field of
\emph{discrepancy theory}. This area studies the \emph{discrepancy} of 
set systems, i.e. the maximum difference of ``red'' and ``blue'' elements 
in any set for the best 2-coloring.
Formally, the discrepancy of a set system $\setS \subseteq 2^{[n]}$ is defined as
\[
\disc(\setS) = \min_{\chi : [n] \to \{ ± 1\}} \max_{S\in \setS} |\chi(S)|. 
\]
In fact, for a variety of problems, the entropy method is the only known technique 
to derive the best bounds (see e.g.~\cite{SixStandardDeviationsSuffice-Spencer1985, DiscrepancyOfPermutations-SpencerEtAl}). 

\subsection{Related work}

Most approximation algorithms that aim at rounding a fractional solution to
an integral one, use one of the following common techniques: 
A classical application of the \emph{properties of basic solutions} yields a $2$-approximation for 
{\sc Unrelated Machine Scheduling}~\cite{UnrelMachineScheduling-LenstraShmoysTardos-FOCS87}.
\emph{Iterative rounding} was e.g. used in a $2$-approximation for a wide class of 
network design problems, like {\sc Steiner Network}~\cite{IterativeRounding-2apx-for-Steiner-Network-JainFOCS98}, 
\emph{randomized rounding} can be used for a $O(\log n/\log \log n)$-approximation for
{\sc Min Congestion}~\cite{RaghavanThompsonRandomizedRounding87} or {\sc Atsp}~\cite{ATSP-logN-loglogN-apx-Goemans-SODA2010}. 
A combination of both techniques provides the currently best approximation
guarantee for {\sc Steiner Tree}~\cite{SteinerTreeBGRS-STOC2010}.
The \emph{dependent rounding} scheme was successfully applied to LPs of an assignment type~\cite{DependentRoundingGandhiKhullerSrinivasan-JACM06}. 
Sophisticated probabilistic techniques like the \emph{Lov{á}sz Local Lemma} were
for example used to obtain $O(1)$-approximation for the {\sc Santa Claus} problem~\cite{Constant-gap-for-Santa-Claus-FeigeSODA08,AspectsConstructiveLLL2010}.

However, to the best of our knowledge, the entropy method has never been 
used for the purpose of approximation algorithms, while being very popular for finding
low discrepancy colorings. 
For the sake of comparison: for a general set system $\setS$ with $n$ elements, 
a random coloring provides an easy bound of $\disc(\setS) \leq O(\sqrt{n \log
  (2|\setS|)})$ (see e.g.~\cite{GeometricDiscrepancy-Matousek99}). 
But using the Entropy method, this can be improved to $\disc(\setS) \leq O(\sqrt{n\log(2|\setS|/n)})$
for $n \leq |\setS|$~\cite{SixStandardDeviationsSuffice-Spencer1985}. This bound is tight, if no more properties on the set system are specified.
Other applications of this method give
a $O(\sqrt{t} \log n)$ bound if no element is in more than $t$ sets~\cite{Discrepancy-Srinivasan-SODA97} and a $O(\sqrt{k}\log n)$ bound for the discrepancy of $k$ permutations. For the first quantity, alternative proof techniques
give bounds of $2t-1$~\cite{IntegerMakingTheorems-BeckFiala81} and $O(\sqrt{t \cdot \log n})$~\cite{BalancingVectors-Banaszczyk98}.
We recommend the book of
 Matou{\v{s}}ek~\cite{GeometricDiscrepancy-Matousek99} (Chapter 4) for
 an introduction to discrepancy theory.


The entropy method itself is purely existential due to the use of the pigeonhole principle.
But in a very recent
breakthrough, Bansal~\cite{DiscrepancyMinimization-Bansal-FOCS2010} showed
how to obtain colorings matching the
Spencer~\cite{SixStandardDeviationsSuffice-Spencer1985} and
Srinivasan~\cite{Discrepancy-Srinivasan-SODA97} bounds,
by considering a random walk guided by the solution of a semidefinite
program.

\subsection*{Our contributions}

In this work, we present a very general rounding theorem which for a given
vector $x \in [0,1]^m$, matrices $A$ and $B$, weights $\mu_i$ 
and an objective function $c$, computes a binary random vector $y$
which (1) preserves all expectations; (2) guarantees worst case bounds
on  $|A_ix - A_iy|$  and  $|B_ix - B_iy|$ and (3) provides strong tail bounds.
The bounds for $A$ depend on the entropy of random functions that approximately
describe the outcomes of random colorings of subsets of columns of $A$, 
while the bounds for rows of $B$ are functions of the weights $\mu_i$.

We use this rounding theorem to obtain better approximation guarantees for
several well studied  {\sc Bin Packing} generalizations. 
In fact, so far all asymptotic FPTAS results for {\sc Bin Packing} related problems
in the literature are based on 
rounding a 
basic solution to a column-based LP using its sparse support. 
We give the first alternative method to round such LPs, which turns out to
be always at least as good as the standard technique (e.g. for classical {\sc Bin Packing}) and significantly stronger 
for several problems. We demonstrate this by providing the following results: 
\begin{itemize}
\item A randomized polynomial time $OPT + O(\log^2 OPT)$ algorithm 
for {\sc Bin Packing With Rejection}, where in contrast to classical {\sc Bin Packing},
each item can either be packed into a bin or rejected at a given cost. 
Our result improves over the previously best 
bound of $OPT + \frac{OPT}{(\log OPT)^{1-o(1)}}$~\cite{AFPTASforVariantsOfBinPacking-EpsteinLevin09}.
\item We give the first (randomized) AFPTAS for the {\sc Train Delivery} problem,
which is a combination of a one-dimensional
 vehicle routing problem and {\sc Bin Packing}. In fact, our algorithm produces solutions 
of cost $OPT + O(OPT^{3/5})$ (see \cite{TrainDelivery-DasMathieuMozes-WAOA2010} for an APTAS).
\end{itemize}
It would not be difficult to extend this list with further variants\footnote{Some examples: 
In {\sc Generalized Cost Variable Size Bin Packing} a list of bin types $j=1,\ldots,k$, each one with
individual cost $c_j \in [0,1]$ and capacity $b_j \in [0,1]$ is given 
 (see \cite{APTASGenCostVariableSizeBinPacking-Epstein08} for an APTAS). 
We can obtain a $OPT + O(\log^2 n)$ approximation. In its well-studied special case
of {\sc Variable Size Bin Packing} the bin costs equal the bin capacities (i.e. $c_j=b_j$ for all $j$)
and we can refine the bound to $OPT + O(\log^2 OPT)$ (see \cite{AFPTAS-for-Variable-Size-BinPacking-Murgolo1987}
for an AFPTAS). 
For {\sc Bin Packing With Cardinality Constraints}, no bin may receive more than $K$ 
items~\cite{BinPackingWithRejection-AFPTAS-EpsteinLevin09}. We can get an $OPT + O(\log^2 n)$ 
approximation. However, we postpone proofs of this claims to the full version.},
but we also believe that the method will find applications that are not related to  {\sc Bin Packing}.

\subsection*{Organization}

We recall some tools and notation in Section~\ref{sec:Preliminaries}. 
In Section~\ref{sec:DiscTheoryRevisited} we revisit results from discrepancy
theory and modify them for our purposes. In Section~\ref{sec:RoundingTheorem}
we show our general rounding theorem.
Then in Sections~\ref{sec:BinPackingWithRejection} and \ref{sec:TrainDelivery} we 
demonstrate how our rounding theorem can be used to obtain approximation
algorithms.
In the Appendix we provide details on how to turn the existential proofs into 
polynomial time algorithms using semidefinite programming and how to solve the
presented LP relaxations in polynomial time. 

\section{Preliminaries\label{sec:Preliminaries}}

The \emph{entropy} of a random variable $Z$ is defined as
\[
  H(Z) = \sum_x \Pr[Z=x]\cdot \log_2\left( \frac{1}{\Pr[Z=x]}\right)
\]
Here the sum runs over all values that $Z$ can attain. Imagine that a data source generates
a string of $n$ symbols according to distribution $Z$. Then 
intuitively, an optimum compression needs asymptotically 
for $n\to \infty$ an expected number of $n\cdot H(Z)$ many bits to encode the 
string. Two useful facts on entropy are:
\begin{itemize}
\item \emph{Uniform distribution maximizes entropy:} If $Z$ attains $k$ distinct values, then
$H(Z)$ is maximal if $Z$ is the uniform distribution. In that case $H(Z) = \log_2(k)$.
Conversely, if $H(Z) \leq \delta$, then there must be at least one event $x$ with $\Pr[Z = x] \geq (\frac{1}{2})^{\delta}$.
\item \emph{Subadditivity:} If $Z,Z'$ are random variables and $f$ is any function, then $H(f(Z,Z')) \leq H(Z) + H(Z')$.
\end{itemize}
We define $H_{\chi \in \{± 1\}^m}f(\chi)$
as the entropy of $f(\chi)$, where $\chi$ is uniformly chosen from $\{ ± 1\}^m$.
See the book of \cite{ProbabilisticMethod-AlonSpencer08} for an intensive
introduction into properties of the entropy function.
We will make use of the \emph{Azuma-Hoeffding Inequality} 
(see e.g. Theorem 12.4 in \cite{ProbabilityAndComputingMitzenmacherUpfal05}).
\begin{lemma} \label{lem:AzumaHoeffdingInequality}
Let $X_1,\ldots,X_n$ be random variables with $|X_i| \leq \alpha_i$ and $E[X_i \mid X_1,\ldots,X_{i-1}]=0$
for all $i=1,\ldots,n$. 
Let $X := \sum_{i=1}^n X_i$. Then 
$\Pr[|X| \geq \lambda\|\alpha\|_2] \leq 2e^{-\lambda^2/2}$
 for any $\lambda \geq 0$.
This still holds, if the distribution of $X_i$ is an arbitrary function of $X_1,\ldots,X_{i-1}$.
\end{lemma}
The sequence $X_1,\ldots,X_n$ is called a \emph{Martingale} and the $\alpha_i$'s
are the corresponding \emph{step sizes}. 
Another tool that we are going to use is a special case of the so-called
\emph{Isoperimetric Inequality} of Kleitman~\cite{IsoperimetricInequality-Kleitman66}.
\begin{lemma} \label{lem:IsoperimetricInequality}
For any $X \subseteq \{0,1\}^m$ of size $|X| \geq 2^{0.8m}$ and $m\geq2$, there are $x,y\in X$ with $\|x - y\|_1 \geq m/2$.
\end{lemma}

A function $\chi : [m] \to \{ 0, ± 1\}$ is called a \emph{partial coloring}.
If at most half of the entries are $0$, then $\chi$ is called a \emph{half-coloring}.
For a quantity $z \in \setZ$, $\lceil z \rfloor$ denotes the integer that is
closest to $z$ (say in case of a tie we round down). If $z\in \setR^m$,
then $\lceil z \rfloor = (\lceil z_1\rfloor,\ldots,\lceil z_m \rfloor)$.
For a matrix $A \in \setR^{n × m}$ and $J \subseteq \{ 1,\ldots,m\}$, $A^J$ denotes the submatrix containing only the columns indexed in $J$.
A submatrix $A' \subseteq A$ will always correspond to a subset of columns of $A$, i.e. $A' \in \setR^{n × m'}$
with $m' \leq m$. If $\chi : J \to \setR$ is only defined on a subset $J \subseteq \{1,\ldots,m\}$ and we write $A\chi$, then we implicitly fill the undefined entries
in $[m] \backslash J$ with zeros.
We say that an entry $x_i \in [0,1[$ has a \emph{finite dyadic expansion} with $K$
bits, if there is a sequence $b_1,\ldots,b_K \in \{ 0,1\}$ with $x_i = \sum_{k=1}^K 2^{-k}\cdot b_{k}$.


\section{Discrepancy theory revisited\label{sec:DiscTheoryRevisited}}

Initially the entropy method was developed to find a coloring $\chi : [m] \to \{ ± 1\}$  
minimizing $|\sum_{i\in S} \chi_i|$ for all sets in a set system, or equivalently to color
columns of the incidence matrix $A \in \{ 0,1\}^{n × m}$ of the set system 
in order to minimize $\|A\chi\|_{\infty}$. 
In contrast, in our setting the matrix $A$ can have arbitrary entries, but
the main technique still applies.

\begin{theorem} \label{thr:ExistanceHalfColoring}
Let  $A \in \setR^{n × m}$ be a matrix with parameters $\Delta_1,\ldots,\Delta_n > 0$ such that
\[
  \mathop{H}_{\bar{\chi} \in \{ ± 1\}^m} \left(\left\{ \left\lceil\frac{A_i\bar{\chi}}{2\Delta_i}\right\rfloor \right\}_{i=1,\ldots,n}\right) \leq \frac{m}{5}
\]
Then there exists a half-coloring $\chi : [m] \to \{± 1, 0\}$ with
$|A_i\chi| \leq \Delta_i$ for all $i=1,\ldots,n$.
\end{theorem}
\begin{proof}
From the assumption, we obtain that there must be a  $b \in \setZ^n$ such that
\[
  \Pr_{\bar{\chi} \in \{ ± 1\}^m} \left[\left\{ \left\lceil\frac{A_i\bar{\chi}}{2\Delta_i}\right\rfloor \right\}_{i=1,\ldots,n} = b\right] \geq \left(\frac{1}{2}\right)^{m/5}
\]
In other words there is a subset $\setY \subseteq \{ ± 1\}^m$ of at least $2^m \cdot (\frac{1}{2})^{m/5} = 2^{\frac{4}{5}m}$ colorings
such that $ \left\lceil\frac{A_i\chi}{2\Delta_i}\right\rfloor = b_i$ for all $\chi \in \setY$ and $i=1,\ldots,n$. 
The Isoperimetric Inequality (Lemma~\ref{lem:IsoperimetricInequality}) 
then yields
the existence of $\chi',\chi'' \in \setY$ with $|\{ j \mid \chi'_j \neq \chi''_j\}| \geq m/2$.
We choose $\chi_j := \frac{1}{2}(\chi_j' - \chi_j'')$, then $\chi \in \{ 0,± 1\}^m$ is the desired half-coloring.
Finally, let us inspect the discrepancy of $\chi$:  $\left|A_i\chi \right| \leq \frac{1}{2}\left| A_i\chi' - A_i\chi'' \right| \leq  \Delta_i$.
\end{proof}
The core of this proof was to show that there is an exponential number of 
colorings $\chi',\chi''$ that are \emph{similar}, meaning that $A\chi' \approx A \chi''$.
This was done by considering disjoint intervals of length $2\Delta_i$
(for every $i$) and using entropy to argue that many colorings must fall into
the same intervals. But on the other hand, $A_i\chi'$ and $A_i\chi''$ might be very close
to each other, while they fall into different intervals and $\chi',\chi''$
would not count as being similar. 

Hence we want to generalize the notion of similarity from 
Theorem~\ref{thr:ExistanceHalfColoring}. Let $A \in \setR^{n × m}$ be a matrix and $\Delta = (\Delta_1,\ldots,\Delta_n)$ be a vector
with $\Delta_i > 0$. 
Then we define the $\Delta$-\emph{approximate entropy} of $A$ as\footnote{The minimum is always attained since all probabilities
are multiplies of $(\frac{1}{2})^m$ and consequently the entropy can
attain only a finite number of values.}
\[
  H_{\Delta}(A) := \min_{f_1,\ldots,f_n: \{ ± 1\}^m \to \setR}\left\{ \mathop{H}_{\chi \in \{ ± 1\}^m}(f_1(\chi),\ldots,f_n(\chi)) : \left|A_i\chi - f_i(\chi)\right| \leq \Delta_i \; \forall i=1,\ldots,n\right\}
\]
First of all note that $H_{\Delta}(A)$ is always upper bounded by the 
entropy of the random variables $\left\lceil\frac{A_i\chi}{2\Delta_i}\right\rfloor$, 
since one 
can choose $f_i(\chi) := 2\Delta_i\cdot \left\lceil\frac{A_i\chi}{2\Delta_i}\right\rfloor$. On the other hand, the claim of 
Theorem~\ref{thr:ExistanceHalfColoring} still holds true 
if the assumption is replaced by 
$H_{\Delta}(A) \leq \frac{m}{5}$, since then one has exponentially many colorings $\bm{Y}$ such that the values $f_i(\chi)$ coincide for every $\chi \in \bm{Y}$ and hence for every half-coloring $\chi := \frac{1}{2}(\chi'-\chi'')$ obtained from colorings $\chi',\chi'' \in \bm{Y}$ one has $|A_i\chi| \leq \frac{1}{2}|(A_i\chi'-f_i(\chi')) - (A_i\chi''-f_i(\chi''))| \leq \Delta_i$. More formally:
\begin{corollary} \label{cor:ExistanceHalfColoringForLazyEntropy}
Let  $A \in \setR^{n × m}$, $\Delta := (\Delta_1,\ldots,\Delta_n) > \mathbf{0}$ 
with $H_{\Delta}(A) \leq \frac{m}{5}$. Then there exists a half-coloring 
$\chi : [m] \to \{± 1, 0\}$ with $- \Delta \leq A\chi \leq \Delta$.
\end{corollary}
Moreover, also $H_{\Delta}$ is subadditive, i.e. $H_{(\Delta,\Delta')}([\begin{smallmatrix} A \\ B \end{smallmatrix}]) \leq H_{\Delta}(A) + H_{\Delta'}(B)$, which follows directly
from the subadditivity of the entropy function (here $[\begin{smallmatrix} A \\ B \end{smallmatrix}]$ is obtained by stacking matrices $A$ and $B$). 

For now let us consider a concrete method of bounding the entropy of a random variable of 
the form $\left\lceil\frac{\alpha^T\chi}{2\Delta}\right\rfloor$, where $\alpha$ is one of the
row vectors of $A$. Recall that this immediately upperbounds $H_{\Delta}(\alpha)$. 
For this purpose, we again slightly adapt a lemma from discrepancy theory
(see e.g. Chapter~4 in \cite{GeometricDiscrepancy-Matousek99}).
\begin{lemma} \label{lem:EntropyBound}
Let $\alpha \in \setR^{m}$ be a vector and $\Delta > 0$. For  $\lambda = \frac{\Delta}{\|\alpha\|_2}$,
\[
 \mathop{H}_{\chi \in \{ ± 1\}^m}\left(\left\lceil \frac{ \alpha^T\chi}{2\Delta} \right\rfloor\right) \leq
G(\lambda) := \begin{cases} 9e^{-\lambda^2/5} & \textrm{if } \lambda \geq 2 \\
\log_2( 32 + 64/\lambda) & \textrm{if } \lambda < 2 
\end{cases} 
\]
\end{lemma}
The proof can be found in Appendix~\ref{sec:OmittedProofForEntropyBound}.
But the intuition is as follows: Abbreviate $Z:=\big\lceil \frac{ \alpha^T\chi}{2\Delta} \big\rfloor$. 
Then $\Pr[Z=0] \geq 1 - e^{-\Omega(\lambda^2)}$ and $\Pr[Z=i] \leq e^{-\Omega(i^2\lambda^2)}$ for $i \neq 0$. 
A simple calculation yields that $H(Z) \leq e^{-\Omega(\lambda^2)}$. 
But for $\lambda \ll 2$, with high probability one has at least  $|Z| \leq O(\frac{1}{\lambda})$ and 
consequently $H(Z) \leq \log O(\frac{1}{\lambda})$.


The following function $G^{-1}(b)$ will denote the discrepancy bound $\Delta$ that
we need to impose, if we do not want to account an entropy contribution of more than $b$.
\[
  G^{-1}(b) := \begin{cases} \sqrt{ 10\ln\left(\frac{9}{b}\right) } & 0 < b \leq 6 \\
 128\cdot(\frac{1}{2})^b & b > 6
  \end{cases}
\]
Strictly spoken, $G^{-1}$ is not the inverse of $G$, but it is not difficult to verify 
that $G(G^{-1}(b)) \leq b$ for all $b>0$.
In other words, for any vector $\alpha$ and value $b>0$, we can choose 
$\Delta := G^{-1}(b) \cdot \| \alpha \|_2$, then $H\left( \Big\lceil \frac{\alpha^T\chi}{2\Delta} \Big\rfloor\right) \leq b$.

\section{The main theorem\label{sec:RoundingTheorem}}

Now we have all ingredients for our main theorem, in which we iteratively
round a fractional vector $x$ using half-colorings $\chi$.
Concerning the choice of parameters $\Delta$, one has in principle two
options: One can either give static bounds  $\Delta_i$ to rows $A_i$ such
that $H_{\Delta}(A') \leq \frac{\#\textrm{col}(A')}{5}$ holds for any submatrix $A' \subseteq A$;
or one can assign a fixed fraction to each row and 
then letting $\Delta_i$ be a function of  $\#\textrm{col}(A')$. 
In fact, we will combine these approaches, which will turn out to be useful later.

\begin{theorem} \label{thm:GeneralRoundingTheorem}
Assume the following is given: A matrix $A \in \setR^{n_A × m}$, parameters  $\Delta = (\Delta_1,\ldots,\Delta_{n_A}) > \mathbf{0}$ such that
$\forall J \subseteq \{ 1,\ldots,m\}: H_{\Delta}(A^J) \leq \frac{|J|}{10}$, a matrix $B \in [-1,1]^{n_B × m}$, weights
$\mu_1,\ldots,\mu_{n_B}>0$ with $\sum_{i=1}^{n_B} \mu_i \leq 1$, a vector $x \in [0,1]^m$ and an objective function $c\in [-1,1]^m$.
Then there is a random variable $y\in\{0,1\}^{m}$ with
\begin{itemize}
\item Preserved expectation:  $E[c^Ty] = c^Tx$, $E[Ay]=Ax$, $E[By] = Bx$.
\item Bounded difference: $|c^Tx - c^Ty| \leq O(1)$;
$|A_ix - A_iy| \leq \log( \min\{ 4n,4m\}) \cdot \Delta_i$ for all $i=1,\ldots,n_A$ ($n := n_A + n_B$);
$|B_ix - B_iy| \leq O(\sqrt{1/\,\mu_i})$ for all $i=1,\ldots,n_B$.
\item Tail bounds: $\forall i: \forall \lambda \geq 0$: $\Pr[|A_ix - A_iy| \geq \lambda \cdot \sqrt{\log (\min\{ 4n,4m\})}\cdot \Delta_i] \leq 2e^{-\lambda^2/2}$.
\end{itemize}
\end{theorem}

\begin{proof}
First, observe that we can append the objective function as an 
additional row to matrix $B$ (with a weight of say $\mu_c := \frac{1}{2}$ and halving the other $\mu_i$'s), 
and so we ignore it from now on.
Next, consider the linear system 
\begin{eqnarray*}
Az &=& Ax \\
Bz &=& Bx \\ 
0 \leq z_{j} &\leq& 1 \quad \forall j=1,\ldots,m
\end{eqnarray*}
and let $z$ be a basic solution. 
Apart from the $0/1$ bounds, the system has only $n$ constraints, hence the number of entries $j$
with $0 < z_j < 1$ is bounded by $n$. 
One can remove  columns of $A$ with $z_j \in \{ 0,1\}$ and apply the
Theorem to the residual instance. 
Hence we set $x := z$ and assume from now on that $m \leq n$.

Furthermore we assume that $x$ has a finite dyadic expansion, i.e. 
every entry $x_j$ it can be written in binary encoding with $K$ bits, for some $K \in \setN$. This can
be achieved by randomly rounding the entries of $x$ to either the
nearest larger or smaller multiple of $(\frac{1}{2})^K$ for a polynomially large $K$, while the
error is exponentially small in $K$.\footnote{Note that we have the term $\min\{ 4m,4n\}$ instead of
$\min\{ 2m,2n\}$ in the claim, to account for the rounding error to obtain a vector $x$ 
 with dyadic expansion and to account for the extra row $c$ that we appended to $B$.}  
We perform the following rounding procedure:
\begin{enumerate}
\item[(1)] WHILE $x$ not integral DO
  \begin{enumerate}
  \item[(2)] Let $k \in \{ 1,\ldots,K\}$ be the index of the least value bit in
any entry of $x$
  \item[(3)] $J := \{ j \in \{1,\ldots,m\} \mid x_j\textrm{'s }k\textrm{th bit is }1\} $
  \item[(4)] Choose $\chi \in \{0,± 1\}^{m}$ with $\chi(j)=0$ for $j\notin J$,
$|\supp(\chi)| \geq |J|/2$, $|A_i\chi| \leq \Delta_i$ and $|B_i\chi| \leq G^{-1}(\mu_i|J|/10)\cdot \sqrt{|J|}$ for all $i$. 
  \item[(5)] With probability $\frac{1}{2}$, flip all signs in $\chi$
  \item[(6)] Update $x := x + (\frac{1}{2})^k\chi$
  \end{enumerate}
\end{enumerate}
The interval of iterations in which bit $k$ is rounded, is termed \emph{phase }$k$.
Let $x^{(k)}$ be the value of $x$ at the beginning of phase $k$ 
and let $x^{(k,t)}$ denote the value of $x$ 
at the beginning of the $t$th to last iteration of phase $k$. 
From now on $x$ always denotes the initial value, i.e. $x = x^{(K)}$ and our choice for
the rounded vector is $y := x^{(0)}$.


Observe that flipping the signs in step (5) ensures that the expectations are preserved, i.e. $E[Ay]=Ax$ and $E[By]=Bx$.
There are two
main issues: $(I)$ showing that the choice of $\chi$ in step (4) is always possible; $(II)$ bounding
the rounding error of $y$ w.r.t. $x$. 
\begin{claimI} For any $J \subseteq \{1,\ldots,m\}$ there is a $\chi \in \{0,± 1\}^{m}$ with $\chi(j)=0$ for $j\notin J$,
$|\supp(\chi)| \geq |J|/2$, $|A_i\chi| \leq \Delta_i$ and $|B_i\chi| \leq G^{-1}(\frac{\mu_i|J|}{10})\cdot \sqrt{|J|}$ for all $i$. 
\end{claimI} 
\begin{proofofclaim}
Our aim is to apply Theorem~\ref{thr:ExistanceHalfColoring} to the stacked $n × |J|$
matrix $\tilde{A} = \big[\begin{smallmatrix} A^J \\ B^j \end{smallmatrix}\big]$ with 
parameter $\tilde{\Delta} := (\Delta,\Delta')$ and  $\Delta_i' := G^{-1}(\frac{\mu_i|J|}{10})\cdot \sqrt{|J|}$. Note that $\|B_i^J\|_2 \leq \sqrt{|J|}$
since $B$ has entries in $[-1,1]$, hence the entropy that we need to account to
the $i$th row of $B$ is $H\left(\left\lceil \frac{B_i\chi}{2\Delta_i'} \right\rfloor\right) \leq \frac{\mu_i |J|}{10}$.
By subadditivity of the (approximate) entropy function and the assumption that $H_{\Delta}(A^J) \leq \frac{|J|}{10}$, 
\[
H_{\tilde{\Delta}}(\tilde{A}) \leq 
H_{\Delta}(A^J) + 
\sum_{i=1}^{n_B} \mathop{H}_{\chi \in \{ ± 1\}^J} \left(\left\lceil\frac{B_i^J\chi}{2\Delta_i'} \right\rfloor\right) \leq \frac{|J|}{10} + \sum_{i=1}^{n_B} \frac{\mu_i|J|}{10} \leq \frac{|J|}{5}.
\]
Thus the requirements of Theorem~\ref{thr:ExistanceHalfColoring} are met, which then
implies the existence of the desired half-coloring and Claim $(I)$ follows. 
\end{proofofclaim}

The next step is to bound the rounding error.

\begin{claimII}  
One has
$|A_ix - A_iy| \leq \log(2m)\cdot \Delta_i$ for all $i=1,\ldots,n_A$
and
$|B_ix - B_iy| \leq O(\sqrt{1/\,\mu_i})$ for all $i=1,\ldots,n_B$.
\end{claimII}
\begin{proofofclaim}
Let $J(k,t) = \{ j \mid x^{(k,t)}_j\textrm{'s }k\textrm{th bit is }1\}$ denote the set $J$ in the $t$th to last iteration of phase $k$, i.e.
$J(k,t) \supset J(k,t-1)$ for any $t$.
Since the cardinality of $J(k,t)$ drops by a factor of at least $1/2$ from iteration to iteration,
we have $|J(k,t-1)| \leq \frac{1}{2}\cdot|J(k,t)|$ for any $t$. Hence each phase has at most $\log_2(m) + 1$
iterations. Then for any $i=1,\ldots,n_A$
\begin{equation} \label{eq:AyMinusAx}
  |A_iy - A_ix| \leq \Big| \sum_{k=1}^{K} \sum_{t\geq 0} A_i(x^{(k,t)} - x^{(k,t+1)})\Big| \leq 
\sum_{k=1}^K \log(2m) \cdot \left(\frac{1}{2}\right)^k \Delta_i \leq \log(2m) \cdot \Delta_i
\end{equation}
using that $|A_i(x^{(k,t)} - x^{(k,t+1)})| \leq (\frac{1}{2})^k \Delta_i$ and  $\sum_{k\geq 1} (\frac{1}{2})^k = 1$. 
Next, consider 
\begin{eqnarray*}
  |B_ix-B_iy| 
&\leq& \sum_{k=1}^K \left(\frac{1}{2}\right)^k \sum_{t\geq0} G^{-1}\Big(\frac{\mu_i|J(k,t)|}{10}\Big)\cdot \sqrt{|J(k,t)|} \\
&\stackrel{(*)}{\leq}&  \sum_{z\in \setZ} G^{-1}(6\cdot 2^z)\cdot \sqrt{2^{z+1} \frac{60}{\mu_i}} \\
&\stackrel{\textrm{Def }G^{-1}}{\leq}& \sqrt{\frac{120}{\mu_i}} \Bigg[ \sum_{z\geq 0} 128 \left(\frac{1}{2}\right)^{2^{z}}2^{z/2} + \sum_{z \geq 0} \sqrt{10\ln\left(\frac{9}{6}2^z\right)} \cdot \left(\frac{1}{2}\right)^{z/2} \Bigg] \\
&\stackrel{(**)}{\leq}& O\big(\sqrt{1/\mu_i}\big)
\end{eqnarray*}
In $(*)$ we use that since $|J(k,t+1)| \geq 2\cdot|J(k,t)|$, for any $k$ and $z$, there is at most 
one $t$ such that $6\cdot 2^z \leq \frac{\mu_i|J(k,t)|}{10} < 6\cdot 2^{z+1}$;
 $(**)$ follows from the convergence of $\sum_{z\geq0} (1/2)^{2^{z}-z/2}$ and $\sum_{z\geq0} \sqrt{z \cdot(1/2)^{z}}$. 
\end{proofofclaim}

Inspecting \eqref{eq:AyMinusAx} again, we see that $A_ix - A_iy$ is a Martingale
and the step size in iteration $t \leq \log(2m)$ of phase $k$ is 
bounded by $\alpha_{k,t} := (\frac{1}{2})^{k} \Delta_i$. Observe that $\|\alpha\|_2 \leq \Delta_i \sqrt{\log_2 (2m)}$,
hence the tail bound $\Pr[|A_ix - A_iy| \geq \lambda \cdot \sqrt{\log(2m)}\cdot \Delta_i] \leq 2e^{-\lambda^2/2}$ for all $\lambda \geq 0$ follows from the Azuma-Hoeffding Inequality (Lemma~\ref{lem:AzumaHoeffdingInequality}).
This concludes the proof of the theorem.
\end{proof}
Moreover, we can also compute such a vector $y$ as guaranteed by the theorem 
in polynomial time, with the only exception that the guaranteed bound on $|B_ix-B_iy|$ is slightly weaker. But we 
still can provide that $E[|B_ix - B_iy|] = O(\sqrt{1/\mu_i})$, 
which is already sufficient for our applications.
We postpone the algorithmic details to Appendix~\ref{sec:Computation}.

For one example application, let  $B \in \{ 0,1\}^{n×n}$ be the incidence matrix 
of a set system with $n$ sets on a ground set of $n$ elements. 
Then apply Theorem~\ref{thm:GeneralRoundingTheorem} with $A=\mathbf{0}$, $x=(\frac{1}{2},\ldots,\frac{1}{2})$
and $\mu_i = \frac{1}{n}$ to obtain a $y \in \{ 0,1\}^n$ with $\|Bx-By\|_{\infty} = O(\sqrt{n})$. 
The coloring $\chi \in \{ 0,1\}^n$ with $y = x + \frac{1}{2}\chi$ is then an $O(\sqrt{n})$
discrepancy coloring, matching the bound of Spencer~\cite{SixStandardDeviationsSuffice-Spencer1985}.
Note that no proof using a different technique is known for Spencer's theorem. 
Hence it seems unlikely that Theorem~\ref{thm:GeneralRoundingTheorem} (in particular the dependence on $1/\mu_i$) 
could be achieved by standard techniques (such as using properties
of basic solutions or the usual independent randomized rounding).

\section{Application: Bin Packing with Rejection\label{sec:BinPackingWithRejection}}

For classical {\sc Bin Packing}, the input consists of a list of \emph{item sizes} $1 \geq s_1 \geq \ldots \geq s_n > 0$
and the goal is to assign the items to a minimum number of \emph{bins} of size $1$. 
For the performance of heuristics like \emph{First Fit}, \emph{Next Fit} and \emph{First Fit Decreasing},
 see \cite{Johnson73,JohnsonFFD74,BinPackingSurvey84}. A proof of strong $\mathbf{NP}$-hardness can be found
in \cite{GareyJohnson79}.
Fernandez de la Vega and Luecker~\cite{deLaVegaLueker81} developed  an 
asymptotic polynomial time approximation scheme (APTAS). Later, Karmarkar and Karp~\cite{KarmarkarKarp82} (see also \cite{KorteVygen02, DesignApproximationAlgorithmsShmoysWilliamson11})
found an algorithm that needs at most $O(\log^2 OPT)$ bins more than 
the optimum solution. 
%

In this section, we provide an application of our  Entropy Rounding Theorem to
the more general problem of {\sc Bin Packing With Rejection}, where
every item $i$ can either be packed into a unit cost bin or it can be
rejected at cost $\pi_i > 0$ (which is also part of the input).
The first constant factor approximation and online algorithms were studied in~\cite{BinPackingWithRejectionPenaltiesDosaHe2006}.
Later an asymptotic PTAS was developed by \cite{BinPackingWithRejectionAPTAS-Epstein-WAOA2006,BinPackingWithRejectionAPTAS-EpsteinJournal2010} (see~\cite{FastAPTASforBinPackingWithRejection-BeinCorreaXin08} for a faster APTAS).
Recently Epstein \& Levin~\cite{AFPTASforVariantsOfBinPacking-EpsteinLevin09} found an algorithm with running time polynomial in the input length and $\frac{1}{\varepsilon}$ which provides solutions of
quality
$(1+\varepsilon)OPT + 2^{O(\frac{1}{\varepsilon}\log \frac{1}{\varepsilon})}$ (implying an $APX \leq OPT + \frac{OPT}{(\log OPT)^{1-o(1)}}$ algorithm for an optimum choice of $\varepsilon$).

An \emph{asymptotic FPTAS (AFPTAS)} is defined as an approximation algorithm $APX$ producing solutions 
with $APX \leq (1+\varepsilon) OPT + f(1/\varepsilon)$ in polynomial time (both in the input length and $\frac{1}{\varepsilon}$).
But there is some ambiguity in the literature what concerns the term $f(1/\varepsilon)$.
According to \cite{KorteVygen02} and \cite{HandbookOfApproxAlgosAndMetaheuristics-ChapterOnAPTAS-OCallaghan-Motwani-Zhu-2007},
$f$ can be any function (equivalent to the requirement $APX \leq OPT + o(OPT)$), while Johnson~\cite{NPcompletenessColumn-Johnson85} 
requires $f$ to be bounded by a polynomial (which is equivalent to $APX \leq OPT + O(OPT^{1-\delta})$ for a fixed $\delta>0$). 
However, we will now obtain a polynomial time algorithm for {\sc Bin Packing With Rejection} with $APX \leq OPT + O(\log^2 OPT)$, which satisfies also
the stronger definition of Johnson~\cite{NPcompletenessColumn-Johnson85} and matches the bound for the special case of {\sc Bin Packing} (without rejection).




We define a set system $\setS = \bm{B} \cup \bm{R}$ with potential bin patterns $\bm{B} = \{ S \subseteq[n] \mid \sum_{i\in S} s_i \leq 1\}$ (each set $S \in \bm{B}$ has cost $c_S := 1$) and rejections $\bm{R} = \{ \{ i \} \mid i\in[n] \}$ at cost $c_{\{i\}}:=\pi_i$
for $i\in[n]$. Then a natural column-based LP is
 \begin{equation}  \label{eq:GeneralLP}
OPT_f = \min\Big\{ c^Tx \mid \sum_{S\in\setS} x_S \mathbf{1}_{S} = \mathbf{1}, x \geq \mathbf{0} \Big\}
\end{equation}
where $\mathbf{1}_S \in \{ 0,1\}^n$ denotes the characteristic vector of $S$.
In \cite{SetCoverWithReplacement-EisenbrandKakimuraRothvossSanita-IPCO2011},
the Karmarkar-Karp  technique~\cite{KarmarkarKarp82} was modified 
to obtain a $O(\sqrt{n}\cdot \log^{3/2} n)$ bound on the additive integrality gap
of~\eqref{eq:GeneralLP}. Note that due to the dependence on $n$, 
such a bound does not satisfy the definition of an AFPTAS, and hence is incomparable to the result of \cite{BinPackingWithRejection-AFPTAS-EpsteinLevin09}. 
But since $OPT_f \leq n$, our result improves over both bounds~\cite{BinPackingWithRejection-AFPTAS-EpsteinLevin09,SetCoverWithReplacement-EisenbrandKakimuraRothvossSanita-IPCO2011}. 

Despite the exponential number of variables in LP \eqref{eq:GeneralLP}, 
one can compute a basic solution $x$ with $c^Tx \leq OPT_f + \delta$ in time polynomial in $n$ and
 $1/\delta$~\cite{KarmarkarKarp82} using either the Gr{ö}tschel-Lovász-Schrijver
 variant of the Ellipsoid method~\cite{GLS-algorithm-Journal81} or the
Plotkin-Shmoys-Tardos framework for covering and packing problems~\cite{FractionalPackingAndCovering-PlotkinShmoysTardos-Journal95}. 
Since this fact is rather standard, we postpone details to Appendix~\ref{Appendix:SolvingLPrelaxation}.

In the following we always assume that the items are sorted w.r.t. their sizes 
such that $s_1 \geq \ldots \geq s_n$ and $\pi_i \leq 1$ for all $i=1,\ldots,n$. 
A feasible solution $y \in \{ 0,1\}^{\setS}$ will reserve at least one slot
for every item, i.e. $i$ many slots for items $1,\ldots, i$. The quantity
$i - \sum_{S \in \setS} y_S |S \cap \{1,\ldots,i\}|$, if positive, is called the \emph{deficit} of $\{1,\ldots,i\}$.
It is not difficult to see that if there is no deficit for any of the 
sets $\{1,\ldots,i\}$, then every item can be assigned to a slot -- potentially
of a larger item\footnote{Proof sketch: Assign input items $i$ iteratively in
increasing order (starting with the largest one, i.e. $i=1$) to the smallest available slot. 
If there is none left for item $i$, then there are less then $i$ slots for items $1,\ldots,i$,
thus this interval had a deficit.} (while in case that $y_{S} = 1$ for $S \in \bm{R}$, the slot
for $S = \{ i\}$ would only be used for that particular item).

We term the constraint matrix $P$ of the system \eqref{eq:GeneralLP} the \emph{pattern matrix}. 
Note that some columns of $P$ correspond to bins, others correspond to rejections.
The obvious idea would be to apply our rounding theorem to $P$, but this would 
not yield any reasonable bound. Instead, we define another matrix $A$ of the same format as $P$, 
where $A_i := \sum_{i'=1}^i P_{i'}$ or equivalently, the entries are defined as $A_{iS} = |S \cap \{ 1,\ldots,i\}|$.
The intuition behind this is that if $x$ is a feasible fractional solution then
$Ay - Ax \geq \mathbf{0}$ iff $y$ does not have any deficit.
Indeed, we will apply Theorem~\ref{thm:GeneralRoundingTheorem} to this \emph{cumulated pattern matrix} $A$. As a prerequisite, we need a strong upper bound on the approximate entropy of any submatrix.

\begin{lemma} \label{lem:LazyEntropyBoundCumMatrix}
Let $A \in \setZ_{\geq 0}^{n × m}$ be any matrix in which column $j$ has
non-decreasing entries from $\{0,\ldots,b_j\}$;
let $\sigma = \sum_{j=1}^m b_j$ be the sum over the largest entries in
each column and $\beta := \max_{j=1,\ldots,m} b_j$ be the maximum of those entries. For $\Delta > 0$ one has
\[
  H_{\Delta}(A) \leq O\left(\frac{\sigma \beta}{\Delta^2}\right). 
\]
\end{lemma}
\begin{proof}
We can add rows (and delete identical rows) such that consecutive rows differ in exactly one entry. This
can never lower the approximate entropy. Now we have exactly $n=\sigma$ many rows. 
There is no harm in assuming\footnote{In fact, rounding $\sigma,\beta,\Delta$ to the nearest 
power of 2 only affect the constant hidden in the $O(1)$-notation.} that $\sigma,\Delta$ and $\beta$ 
 are powers of $2$. Let $B \in \setR^{\sigma × m}$ be the matrix with $B_i = A_i - A_{i-1}$ (and $B_1=A_1$). In other words,
$B$ is a $0/1$ matrix with exactly a single one per row.
 
Consider the balanced binary laminar dissection $\setD := \big\{ \{ 2^k(i-1)+1,\ldots,2^ki \} \mid k=0,\ldots,\log \sigma; i=1,\ldots,\frac{\sigma}{2^k} \big\}$ 
of the row indices $\{ 1,\ldots,\sigma\}$.
In other words, $\setD$ contains $2^{k}$ many intervals of length $\sigma/ 2^k$ for $k = 0,\ldots,\log_2 \sigma$.
For every of those interval $D \in \setD$ we define the vector $C_{D} := \sum_{i\in D} B_i \in \setZ^m$ and parameter
$\Delta_D := \frac{\Delta}{32\cdot 1.1^{|z|}}$ if $|D| = 2^z\cdot \frac{\Delta^2}{\beta}$ (with $z \in \setZ$). Note that since $B$ contains a single one per row and at 
most $\beta$ ones per column, thus
$\|C_D\|_{2} \leq (\|C_D\|_{\infty} \cdot \|C_D\|_{1})^{1/2} \leq \sqrt{\beta \cdot |D|} = 2^{z/2}\cdot \Delta$.\footnote{Here we use Hölder's inequality: $\|x\|_2 \leq (\|x\|_{\infty} \cdot \|x\|_{1})^{1/2}$ for every $x \in \setR^m$.}

For every $i$, note that $\{1,\ldots,i\}$ can be written as a disjoint union 
of some intervals $D_1,\ldots,D_{q} \in \setD$, all of different size. 
We choose $f_i(\chi) := \sum_{p=1}^q 2\Delta_{D_p}\cdot \left\lceil \frac{C_{D_p}\chi}{2\Delta_{D_p}} \right\rfloor$. Then 
\[
  \left| A_i\chi - f_i(\chi) \right| \leq \sum_{p=1}^q \Delta_{D_p} \leq \sum_{z \in \setZ} \frac{\Delta}{32\cdot 1.1^{|z|}} \leq \Delta.
\]
Thus $H_{\Delta}(A)$ is upper bounded by the entropy of the random 
variables $f_1(\chi),\ldots,f_n(\chi)$.
But since each $f_i$ is a function of $\left\{\left\lceil C_D\chi/(2\Delta_D) \right\rfloor\right\}_{D \in \setD}$, it in fact suffices to
bound the entropy of the latter random variables.
We remember that for all $z \in \setZ$, one has at most $\frac{\sigma \beta}{2^z \Delta^2}$ many intervals $D \in \setD$ of 
size $|D| = \frac{2^z \Delta^2}{\beta}$, with  $\|C_D\|_{2} \leq  2^{z/2}\cdot \Delta$ and $\Delta_D = \frac{\Delta}{32\cdot 1.1^{|z|}} \geq \frac{\Delta}{32\cdot 1.1^{|z|}\cdot \Delta 2^{z/2}}\|C_D\|_2 \geq \frac{1}{32\cdot 1.3^{z}} \|C_D\|_2$. 
Finally
\begin{eqnarray*}
\mathop{H}_{\chi \in \{ ± 1\}^{m} }\left(\left\{ \left\lceil \frac{C_D\chi}{2\Delta_D} \right\rfloor\right\}_{D \in \setD} \right) 
&\stackrel{\substack{H \textrm{ subadd.} \\ \textrm{\& Lem.~\ref{lem:EntropyBound}}}}{\leq}& \sum_{z\in \setZ} \frac{\sigma \beta}{2^z \Delta^2} \cdot G\left(\frac{1}{32\cdot 1.3^{z}}\right) \nonumber \\
&\stackrel{\textrm{Def. }G}{=}& \frac{\sigma \beta}{\Delta^2}\bigg[ \underbrace{\sum_{-\infty < z< -15} \left(\frac{1}{2}\right)^z 9e^{-(1.3^{-z}/32)^2/5}}_{= O(1)} + \underbrace{\sum_{z \geq -15} \left(\frac{1}{2}\right)^z \cdot \log_2 \left(32 + 64\cdot 32\cdot 1.3^{z}\right)}_{= O(1)}\bigg] \nonumber \\
&=& O\left(\frac{\sigma \beta}{\Delta^2} \right) \nonumber 
\end{eqnarray*}
\qedhere
\end{proof}
The parametrization used in the above proof is inspired by the work of
Spencer, Srinivasan and Tetali~\cite{DiscrepancyOfPermutations-SpencerEtAl}.
A simple consequence of the previous lemma is the following.
\begin{lemma} \label{lem:ApproxEntropyForBinPackingMatrix}
Let $S_1,\ldots,S_m \subseteq [n]$ be a set system with numbers $1 \geq s_1 \geq \ldots \geq s_n > 0$
such that $\sum_{i\in S_j} s_i \leq 1$ for any set $S_j$.
Let $A \in \setZ_{\geq 0}^{n × m}$ be the cumulated pattern matrix, defined by $A_{ij} = |S_j \cap \{ 1,\ldots,i\}|$. 
Then there is a constant $C > 0$ such that for $\Delta := \frac{C}{s_i}$, one has $H_{\Delta}(A) \leq \frac{1}{10} \sum_{j=1}^m \sum_{i \in S_j} s_i \leq \frac{m}{10}$.
\end{lemma}
\begin{proof}
Let $A^{\ell}$ be the submatrix of $A$ consisting of all rows $i$ such that $s_i \geq (1/2)^{\ell}$.
Note that row $A_i$ possibly appears in several $A^{\ell}$.
We apply Lemma~\ref{lem:LazyEntropyBoundCumMatrix} to $A^{\ell}$ with $\beta := \|A^{\ell}\|_{\infty} \leq 2^{\ell}$ and 
$\sigma := \sum_{j=1}^m |S_j \cap \{ i : s_i \geq (1/2)^{\ell}\}|$ to obtain
  $H_{C\cdot 2^{\ell}}(A^{\ell}) \leq \frac{1}{20} \sum_{j=1}^m (1/2)^{\ell}\cdot |S_j \cap \{i : s_i \geq (1/2)^{\ell}\}|$ for $C$ large enough. Eventually
\[
 H_{\Delta}(A) \stackrel{\textrm{subadd.}}{\leq} \sum_{\ell \geq 0} H_{C\cdot 2^{\ell}}(A^{\ell}) \leq \frac{1}{20} \sum_{\ell \geq 0} \sum_{j=1}^m (1/2)^{\ell}\cdot |S_j \cap \{ i : s_i \geq (1/2)^{\ell}\} |  \leq \frac{1}{10} \sum_{j=1}^m \sum_{i \in S_j} s_i  
\]
since item $i \in S_j$ with $(1/2)^{\ell} \leq s_i < (1/2)^{\ell+1}$ contributes at most
$\frac{1}{20} \sum_{\ell' \geq \ell} (1/2)^{\ell'} \leq \frac{1}{10}s_i$ to the left hand side.
\end{proof}

Our procedure to round a fractional ${\textsc{Bin Packing With Rejection}}$ 
solution $x$ will work as follows:
For a suitable value of 
$\varepsilon > 0$, we term all items of size at least $\varepsilon$
\emph{large} and \emph{small} otherwise. We take the cumulated pattern matrix 
$A$ restricted to the large items.  
Furthermore we 
define a matrix $B$ such that $Bx$ denotes the space reserved for small items. Then
we apply Theorem~\ref{thm:GeneralRoundingTheorem} to obtain an integral vector $y$ 
which is then repaired to a feasible solution without significantly increasing
the cost of the solution. 

\begin{theorem} \label{thm:BPRapproximation}
There is a randomized algorithm for {\sc Bin Packing With Rejection} with expected polynomial running time
which produces a solution of cost $OPT_f + O(\log^2 OPT_f)$.
\end{theorem}
\begin{proof}
Compute a  fractional solution $x \in [0,1]^{\setS}$ to the {\sc Bin Packing With Rejection}
LP~\eqref{eq:GeneralLP} with cost $c^Tx \leq OPT_f + 1$ (see Appendix~\ref{Appendix:SolvingLPrelaxation} for details).
We define $\varepsilon := \frac{\log(OPT_f)}{OPT_f}$ 
(assume $OPT_f \geq 2$). For any item $i$ that is rejected in $x$ to a fractional extend of more than
$1-\varepsilon$ (i.e.  $x_{\{i\}} > 1-\varepsilon$), we fully reject item $i$. We
account a multiplicative loss of $1/(1-\varepsilon) \leq 1+2\varepsilon$ (i.e. an additive cost increase of $O(\log OPT_f)$).
From now on, we may assume that $x_S \leq 1-\varepsilon$ for all sets $S \in \bm{R}$.
Let $1,\ldots,L$ be the items of size at least $\varepsilon$, hence $OPT_f \geq \varepsilon^2L$, since every item is covered 
with bin patterns at an extend of at least $\varepsilon$. 
Let $A \in \setZ^{L × \setS}$ with
\[
  A_{iS} := |\{ 1,\ldots,i\} \cap S|
\]
be the cumulated pattern matrix restricted to the large items.
According to Lemma~\ref{lem:ApproxEntropyForBinPackingMatrix}, for a choice of $\Delta_i = \Theta(1/s_i)$,
one has $H_{\Delta}(A') \leq \frac{\# \textrm{col}(A')}{10}$ for every submatrix $A' \subseteq A$.
Choose $B \in [0,1]^{1 × \setS}$ as the row vector 
where $B_{1,S} = \sum_{i \in S: i > L} s_i$ for $S \in \setS$ denotes the space in pattern $S$ that
is reserved for small items. 
%

We apply Theorem~\ref{thm:GeneralRoundingTheoremConstructiveVersion} (Theorem~\ref{thm:GeneralRoundingTheorem} suffices for a non-constructive bound on the integrality gap)
to matrices $A$, $B$ and cost function $c$ 
with $\mu_1 = 1$ to obtain a vector $y \in \{ 0,1 \}^{\setS}$ with the following properties
\begin{itemize}
\item[(A)] The deficit of any interval $\{1,\ldots,i\}$ of large items (i.e. $i\in \{ 1,\ldots,L\}$) is bounded by $O(\frac{1}{s_i}\log L)$. 
\item[(B)] The space for small items reserved by $y$ equals that of $x$ up 
to an additive constant term (formally $|\sum_{S \in \setS} (x_S - y_S) \cdot B_{1,S}| = O(1)$).
\item[(C)] $c^Ty \leq c^Tx + O(1)$
\end{itemize}
For $\ell \geq 0$, we say that the items $G^{\ell} := \{ i \leq L \mid (\frac{1}{2})^{\ell} \geq i > (\frac{1}{2})^{\ell+1}\}$ form \emph{group $\ell$}. Note that at most $\frac{1}{\varepsilon}+1$ groups contain large items.
We eliminate the deficits for large items by packing $O(\log L)$
extra bins with the largest item from every group, hence leading to 
$O(\log L)\cdot O(\log \frac{1}{\varepsilon}) = O(\log^2 OPT_f)$ extra bins. 
%
Property $(B)$ implies that after buying $O(1)$ extra bins for small items, it is possible to assign all small items \emph{fractionally}
to the bought bins (i.e. the small items could be feasibly assigned if it
would be allowed to split them). By a standard argument (see e.g.~\cite{BinPackingWithRejection-AFPTAS-EpsteinLevin09}) this fractional 
assignment can be turned into an integral assignment, if we
 discard at most one item per pattern, i.e. we pack discarded small items of total size at most $\varepsilon\cdot(c^Ty + O(1))$
separately, which can be done with $O(\varepsilon)\cdot OPT_f + O(1) \leq O(\log OPT_f)$ extra bins. 
The claim follows. 
\end{proof}

\section{Application: The Train Delivery Problem\label{sec:TrainDelivery}}

For the {\sc Train Delivery} problem, $n$ items are given as input, but now
every item $i\in\{1,\ldots,n\}$ has a \emph{size} $s_i$ and a \emph{position} $p_i \in [0,1]$. 
The goal is to transport the items to a depot, located at $0$, using
trains of capacity $1$ and minimizing the total tour length\footnote{For definiteness, say the tour must start and
end at the depot and once items are loaded into the train, they have to remain until the depot
is reached.}. 
In other words, it is a combination of one-dimensional vehicle routing and {\sc Bin Packing}.
We define $\setS$ as a set system consisting of all sets $S \subseteq [n]$ 
with $\sum_{i\in S} s_i \leq 1$ and cost $c_S := \max_{i\in S} p_i$ for set $S$. Then the optimum value for 
{\sc Train Delivery} equals the cost of the cheapest subset of $\setS$
covering all items (ignoring a constant factor of two). 

The problem was studied in~\cite{TrainDelivery-DasMathieuMozes-WAOA2010}, where the authors
provide an APTAS. We will now obtain an AFPTAS. First, we make our life easier by 
using a result of~\cite{TrainDelivery-DasMathieuMozes-WAOA2010} saying that modulo a $1+O(\varepsilon)$ factor in the approximation guarantee it suffices to solve
\emph{well-rounded} instances which have the property that
 $\varepsilon \leq p_i \leq 1$ and $p_i \in (1+\varepsilon)^{\setZ}$ for all $i=1,\ldots,n$.
The first condition can be obtained by splitting all tours in a proper way; the second condition is
obtained by simply rounding all positions up to the nearest power of $1+\varepsilon$.
Hence, we can partition the items according to their position by letting $P_j := \{ i \mid p_i = (1+\varepsilon)^j \}$ 
for $j=0,\ldots,t-1$ with $t = O(\frac{1}{\varepsilon} \log \frac{1}{\varepsilon})$.

Our rounding procedure works as follows: analogously to {\sc Bin packing With Rejection}, we construct 
matrices $A(j),B(j)$ separately for the items at each position $j$. Then we stack them together;
apply Theorem~\ref{thm:GeneralRoundingTheorem}, and repair the obtained integral vector to a feasible solution
(again analogous to the previous section). 
Here we will spend a higher weight $\mu_j$
for positions $j$ which are further away from the depot -- since those are costlier to cover.
\begin{theorem}
There is a randomized algorithm with expected polynomial running time 
for {\sc Train Delivery}, providing solutions 
of expected cost $E[APX] \leq OPT_f + O(OPT_f^{3/5})$.
\end{theorem}
\begin{proof}
Compute a fractional solution $x$ for the {\sc Train Delivery} LP (i.e. again LP \eqref{eq:GeneralLP},
but with the problem specific set system and cost vector) of cost $c^Tx \leq OPT_f + 1$ (see Appendix~\ref{Appendix:SolvingLPrelaxation} for details).
We will choose $\varepsilon := 1 / OPT_f^{\delta}$ for some constant $0<\delta<1$ that we determine later
and assume the instance is well-rounded.

By $1,\ldots,L$ we denote the \emph{large} items of size $\geq\varepsilon$. 
Let $A(j) \in \setZ^{(P_j \cap [L]) × \setS}$ with entries $A_{i,S}(j) = |S \cap \{1,\ldots,i\} \cap P_j|$ be the cumulated pattern matrix, 
restricted to large items at position $P_j$.
We equip again every row $A_{i}(j)$ with parameter $\Delta_i(j) := \Theta(1/s_i)$.
Then we stack $A(0),\ldots,A(t-1)$ together to obtain an $L × |\setS|$ matrix $A$.
Again, we need to show that for 
any submatrix $A' \subseteq A$, one has $H_{\Delta}(A') \leq \frac{\#\textrm{col}(A')}{10}$. 
Let $\setS' \subseteq \setS$ be the sets whose characteristic vectors form the columns of $A'$.
We apply Lemma~\ref{lem:ApproxEntropyForBinPackingMatrix} individually to each $A'(j)$ and obtain
\[
  H_{\Delta}(A') \leq \sum_{j=0}^{t-1} H_{\Delta(j)}(A'(j)) \leq \frac{1}{10} \sum_{j=0}^{t-1} \sum_{S \in \setS'} \sum_{i \in S \cap P_j} s_i \leq \frac{\# \textrm{col}(A')}{10}
\]
using that every set $S$ contains items of total size at most $1$. 
%
Furthermore, we define a matrix $B \in [0,1]^{t × \setS}$ with $B_{j,S} := \sum_{i \in S \cap P_j: i > L} s_i$ as the space that
pattern $S$ reserves for small items at position $j$. We equip the $j$th row of $B$ with weight 
$\mu_j := \frac{\varepsilon}{5}\cdot (1+\varepsilon/4)^{-j}$, i.e. the weight grows with the distance to the depot. 
Note that $\sum_{j\geq 0} \mu_j \leq \sum_{j\geq 0} \frac{\varepsilon}{5}\cdot (1+\varepsilon/4)^{-j} = \frac{4}{5} + \frac{\varepsilon}{5} \leq 1$. Moreover 
$OPT_f \geq \sum_{i\in[n]} s_ip_i \geq \varepsilon^2 \cdot L$ (see \cite{TrainDelivery-DasMathieuMozes-WAOA2010}), hence the number of rows of $A$ and $B$ is $L + t \leq \textrm{poly}(OPT_f)$. 

We apply Theorem~\ref{thm:GeneralRoundingTheoremConstructiveVersion} (again Theorem~\ref{thm:GeneralRoundingTheorem} suffices for an integrality gap bound) to obtain an integral vector $y \in \{ 0,1\}^{\setS}$ with the following
error guarantees:
\begin{itemize}
\item \emph{Large items:} Let $i \in P_j, i \leq L$. Then the deficit of $\{ 1,\ldots,i\} \cap P_j$ is bounded by $O(\frac{1}{s_i} \log OPT_f )$. \vspace{1mm} \\
For every position we use $O(\log^2 OPT_f)$ extra bins to eliminate
the deficits of large items, which costs in total $\sum_{j\geq 0} (1+\varepsilon)^{-j}\cdot O(\log^2 OPT_f) = O(\frac{1}{\varepsilon} \log^2 OPT_f)$. 
\item \emph{Small items:} For position $j$, the expected discrepancy in the reserved space for
small items is  $E[|B_jx - B_jy|] \leq O(\sqrt{1/\mu_j}) = \textstyle{O(\sqrt{1/\varepsilon \cdot (1+\varepsilon/4)^{j}})}$. \vspace{1mm} \\ 
We buy $|B_jx - B_jy|$  
extra bins to cover small items at position $j$. 
Their expected cost is bounded by $O(\sqrt{1 / \varepsilon}) \cdot (1+\frac{\varepsilon}{4})^{j/2} \cdot (1+\varepsilon)^{-j} \leq O(\sqrt{1/\varepsilon})\cdot (1-\varepsilon/2)^j$. In total, this accounts with 
an expected cost increase of $O( \sqrt{1/\varepsilon})\sum_{j\geq 0} (1-\varepsilon/2)^j = O(1)\cdot (1/\varepsilon)^{3/2}$ for all positions.
Now, for every position the space reserved for small items is at least as large as the required space, hence the small items
can be assigned fractionally. Then after discarding at most one small item per pattern, even an integral assignment
is possible. 
We account this with a multiplicative factor of $1/(1-\varepsilon) \leq 1+2\varepsilon$.
\end{itemize}
Summing up the bought extra bins, we obtain a solution $APX$ with
\[
 E[APX] \leq (1+O(\varepsilon))OPT_f + O\left(\frac{1}{\varepsilon} \log^2 OPT_f \right) + 
O((1/\varepsilon)^{3/2}) \leq OPT_f + O(OPT_f^{3/5}).
\]
choosing $\varepsilon := OPT_f^{-2/5}$.
\end{proof}

\paragraph{Acknowledgments.}

The author is grateful to Michel X. Goemans, Neil Olver, Rico Zenklusen, Laura Sanit{à} and Nikhil Bansal
for helpful advice and remarks.
Thanks to Saurabh Ray for marketing advice.

\bibliographystyle{alpha}
\bibliography{entropymethod}

\newpage

\appendix

\section*{Appendix}

\section{Computing low-discrepancy colorings by SDP\label{sec:Computation}}

Observe that the only 
non-constructive ingredient we used 
is the application of the 
pigeonhole principle (with an exponential number of pigeons and pigeonholes)
in Theorem~\ref{thr:ExistanceHalfColoring} and Corollary~\ref{cor:ExistanceHalfColoringForLazyEntropy} to obtain the existence of low-discrepancy
half-colorings. The purpose of this section is to make this claim constructive.
More precisely, we will replace each phase of $\log m$ iteratively found
half-colorings by 
finding a single \emph{full coloring}.

\begin{theorem} \label{thm:ComputabilityOfLowDiscrepancyColoring}
Let $A \in \setQ^{n_A × m}$, $B \in [-1,1]^{n_B × m}$, $\Delta_1,\ldots,\Delta_{n_A} > 0$, $\mu_1,\ldots,\mu_{n_B} > 0$ 
(with $\sum_{i} \mu_i \leq 1$, $m \geq 2$) be given as input. Assume that $\forall A' \subseteq A: H_{\Delta}(A') \leq \frac{\#\textrm{col}(A')}{10}$.
Then there is a constant $C>0$ and a randomized algorithm with expected polynomial running time, 
which computes a coloring $\chi : [m] \to \{ ± 1\}$ such that for all $\lambda \geq 0$, 
\begin{eqnarray*}
 \Pr\left[|A_i\chi| > \lambda\cdot C\sqrt{\log m} \cdot\Delta_i\right] &\leq& 4e^{-\lambda^2/2}  \quad\forall i=1,\ldots,n_A \vspace*{2mm} \\
 \Pr\left[|B_i\chi| > \lambda \cdot C\sqrt{1/\mu_i}\right] &\leq& 160\cdot 2^{-\lambda/6}  \quad \forall i=1,\ldots,n_B.
\end{eqnarray*}
\end{theorem}
Note that $\Delta_i \geq \|A_{i}\|_{\infty}$, thus we may rescale $A$ and $\Delta$ so that 
$\Delta_i \geq 1$ and $\|A\|_{\infty} \leq 1$.
We may assume that $n\geq m$. Furthermore we assume $\lambda \geq 1$  and $m$ is large enough, since otherwise all probabilities 
exceed 1 for $C$ suitable large and there is nothing to show. 
The following approach to prove Theorem~\ref{thm:ComputabilityOfLowDiscrepancyColoring} is a 
adaptation of the seminal work of Bansal~\cite{DiscrepancyMinimization-Bansal-FOCS2010}.

\subsection{Some preliminaries}

The \emph{Gaussian distribution} $N(\mu,\sigma^2)$ with 
mean $\mu$ and variance $\sigma^2$ is defined by the density function
\[
  f(x) = \frac{1}{\sqrt{2\pi}\sigma} e^{-(x-\mu)^2/(2\sigma^2)}
\]
If $g$ is drawn from this distribution, we write $g \sim N(\mu,\sigma^2)$.
The $n$-dimensional Gaussian distribution $N^n(0,1)$ 
is obtained by sampling
every coordinate $g_i$ independently from $N(0,1)$. Since  $N^n(0,1)$ is rotationally symmetric, one has

\begin{fact} \label{fact:ScalarVectorGaussian}
Let $v\in\setR^n$ be any vector and $g \sim N^n(0,1)$, then $g^Tv \sim N(0,\|v\|_2^2)$.
\end{fact}

\subsubsection*{Martingales \& concentration bounds}

A \emph{Martingale} is a sequence $0=X_0, X_1, \ldots,X_{n}$ of random variables with the
property that the increment $Y_i := X_i - X_{i-1}$ has mean $E[Y_i] = 0$. Here 
$Y_i := Y_i(X_0,\ldots,X_{i-1})$ is allowed to arbitrarily depend on the previous events
$X_0,\ldots,X_{i-1}$. We will make use of the following concentration bound:

\begin{lemma}[\cite{DiscrepancyMinimization-Bansal-FOCS2010}]  \label{lem:MartingaleBound}
Let $0=X_0,\ldots,X_{n}$ be a Martingale with increments $Y_i$, where $Y_i = \eta_iG_i$, $G_i \sim N(0,1)$
and $|\eta_i| \leq \delta$. Then for any $\lambda \geq 0$
\[
  \Pr[|X_n| \geq \lambda \delta\sqrt{n}] \leq 2e^{-\lambda^2/2}.
\]
\end{lemma}

\subsubsection*{Semidefinite programming}

A matrix $Y \in \setR^{n×n}$ is termed \emph{positive-semidefinite} (abbreviated by $Y \succeq 0$), 
if $x^TYx \geq 0$ for all $x\in\setR^n$ (or equivalently all eigenvalues of $Y$ are non-negative).
Let $S_n = \{ Y \in \setR^{n×n} \mid Y^T=Y; \; Y \succeq 0\}$ be the convex cone of symmetric positive-semidefinite
matrices. 
A \emph{semidefinite program} is of the form
\begin{eqnarray*}
\max \; C \bullet Y & & \\
D^{\ell} \bullet Y &\leq& d_{\ell}  \quad \forall\ell=1,\ldots,k \\
Y &\in& S_n
\end{eqnarray*}
Here $C \bullet Y = \sum_{i=1}^n \sum_{j=1}^n C_{ij}\cdot Y_{ij}$ is the ``vector product for matrices'' (also called
\emph{Frobenius product}).  
In contrast to linear programs, it is possible that the only feasible
solution to an SDP is irrational even if the input is integral. Furthermore the
norm even of the smallest feasible solution might be doubly-exponential
in the input length~\cite{ExactDualityTheoryForSDP-Ramana95}. 
Nevertheless, given an error parameter $\varepsilon > 0$ and
a ball of radius $R$ that contains at least one optimum solution (of value $SDP$), one
can compute a $Y \in S_n$ with $C \bullet Y \geq SDP - \varepsilon$ and $D^{\ell} \bullet Y \leq d_{\ell} + \varepsilon$ for all $\ell = 1,\ldots,k$ in
time polynomial in the input length and in $\log(\max\{ 1/\varepsilon,R \})$. 
Since for our algorithm, numerical errors could be easily absorbed 
into the discrepancy bounds,
we always assume we have exact solutions. 
The first use of SDPs in approximation algorithms was the {\sc MaxCut}
algorithm of Goemans and Williamson~\cite{MaxCut-GoemansWilliamsonJCSS04}. Later on
SDPs were used for example to approximate graph colorings~\cite{GraphColoringBySDP-KargerMotwaniSudan98}.
We refer to the surveys of~\cite{SDPsAndCombOpt-Lovasz03, SDPsInCombOpt-Goemans97}
for more details on semidefinite programming.

Using that any symmetric, positive semidefinite matrix $Y$ can be written as $W^TW$ for $W \in \setR^{n × n}$ (and vice versa), 
the above SDP is equivalent to a \emph{vector program} 
\begin{eqnarray*}
  \max \sum_{i=1}^n \sum_{j=1}^n C_{ij}v_iv_j & & \\
  \sum_{i=1}^n \sum_{j=1}^n D^{\ell}_{ij}v_iv_j &\leq& d_{\ell} \quad \forall \ell=1,\ldots,k \\
  v_i &\in& \setR^n \quad \forall i=1,\ldots,n
\end{eqnarray*}

\subsection{The algorithm}

Consider the following semidefinite program\footnote{More precisely this program is \emph{equivalent} to a semidefinite program.}
\begin{eqnarray*}
 \Big\|\sum_{j=1}^m A_{ij} \cdot v_j \Big\|_2 &\leq& \Delta_i \quad \forall i=1,\ldots,n_A \\
 \Big\|\sum_{j=1}^m B_{ij} \cdot v_j \Big\|_2 &\leq& G^{-1}\left(\frac{\mu_i}{10}|J_{t-1}|\right) \cdot \sqrt{|J_{t-1}|} \quad \forall i=1,\ldots,n_B \\
 \sum_{j=1}^m \|v_j\|_2^2 &\geq& \frac{|J_{t-1}|}{2} \\
  \|v_j\|_2 &\leq& 1 \quad \forall j\in J_{t-1}  \\
  v_j &=& \mathbf{0} \quad \forall j\notin J_{t-1} \\
  v_j  &\in& \setR^{m} \quad \forall j=1,\ldots,m
\end{eqnarray*}
Here, $J_{t-1} \subseteq [m]$ denotes the set of \emph{active} variables at the beginning
of step $t$. Initially we have a fractional coloring $\chi^0 = (0,\ldots,0)$ and 
all variables are active, i.e. $J_0 = [m]$.
For a certain number of iterations, we sample an increment $\gamma_t$ using a 
solution from the SDP and add it to $\chi$. If a variable $\chi(j)$ reaches $+1$ or $-1$,
then we freeze it (the variable becomes inactive and is removed from $J_t$).
Note that the proof of Theorem~\ref{thm:GeneralRoundingTheorem} guarantees that
no matter which variables are currently contained in $J_t$, the SDP is always feasible.

Let $s := \frac{1}{n^2\sqrt{8\log(nm)}}$ be a step size and $\ell = 20\cdot \frac{16}{s^2}\log m$ be the number of iterations.
For $t=1,\ldots,\ell$ repeat the following:
\begin{enumerate}
\item[(1)] Compute a solution $\{ v_j \}_{j \in [m]}$ to the SDP
\item[(2)] Sample $g \sim N^{m}(0,1)$
\item[(3)] Update $\gamma_t(j) := s \cdot g^Tv_j$, $\chi^t := \chi^{t-1} + \gamma_t$
\item[(4)] If $\chi^t(j) \in [1-\frac{1}{n^2},1]$ ($[-1,-1+\frac{1}{n^2}]$, resp.) then $\chi^t(j) := 1$ ($-1$, resp.), $j$ becomes \emph{inactive}
\end{enumerate}
Note that $s\cdot g^Tv_j \sim N(0,s^2 \|v_j\|_2^2)$, hence
 $\Pr[\gamma_t(j) > \frac{1}{n^2}] \leq 2\cdot e^{-(\sqrt{8\ln(nm)})^2/2}=\frac{2}{n^4m^4}$ using Lemma~\ref{lem:MartingaleBound}, hence we may assume that $|\chi^t(j)|$ never
exceeds $1$. No constraint $i$ will ever suffer an extra discrepancy of more than
$n\cdot\frac{1}{n^2} \ll \Delta_i$ in Step (4), hence we ignore it from now on. 

It was proven in \cite{DiscrepancyMinimization-Bansal-FOCS2010}, that
with high probability, after the last iteration all variables
are inactive, i.e. $\chi^{\ell} \in \{ ± 1\}^{m}$. 

\begin{lemma}[\cite{DiscrepancyMinimization-Bansal-FOCS2010}] \label{lem:NumVariablesHalfAfter16-over-s2It}
The probability that within $16/s^2$ iterations, the number of active sets decreases by a factor of at least 2
is at least $1/2$.
\end{lemma}

Intuitively the reason is the following: consider a variable $\chi^t(j)$ and suppose for
simplicity that always $\|v_j\|_2=1$. Then, the values $\chi^0(j),\chi^1(j),\ldots$ behave essentially like
an unbiased random walk in which in every step we go either $s$ units to the left or to the right.
In a block of $\frac{1}{s^2}$ steps, with a constant probability we deviate $\frac{1}{s}$ steps from $0$, i.e. $|\chi^t(j)| \geq s\cdot\frac{1}{s} =1$
and $j$ got frozen at some point. Hence the chance that
any of the $m$ variables $j$ is not frozen after $20\log m$ blocks (each of $\frac{16}{s^2}$ iterations)
can be easily bounded by $e^{-10\log m \cdot 0.9^2/3} \leq \frac{1}{m^2}$ using the Chernov bound (e.g. Thm~4.4 in \cite{ProbabilityAndComputingMitzenmacherUpfal05}).

It remains to bound the discrepancy of $\chi^{\ell}$.
Let $\{v_j^t\}_{j \in [m]}$ be the SDP solution and $g_t$ be the random Gaussian vector in step $t$.
Then
\[
  A_i\chi^{\ell} = \sum_{t=1}^{\ell} A_i\gamma_t = \sum_{t=1}^{\ell} \sum_{j=1}^m s A_{ij} g_t^Tv_j^t = \sum_{t=1}^{\ell} g_t^T s\Big( \sum_{j=1}^m A_{ij} v_j^t \Big)
\]
But $s\sum_{j=1}^m A_{ij}v_j^t$ is a vector of length at most $s\cdot\Delta_i$, thus
$A_i\cdot\chi^{\ell}$ is a martingale and we can apply Lemma~\ref{lem:MartingaleBound} with $\delta := s\cdot\Delta_i$ to bound
\[
  \Pr\Big[|A_i \chi^{\ell}| > \lambda \cdot \sqrt{\tfrac{320}{s^2} \log m}\cdot s\cdot\Delta_i \Big] \leq 2e^{-\lambda^2/2 }. 
\]
However, we need to be a bit more careful to analyze the behavior of $|B_i\chi^{\ell}|$.
In the following, we fix any index $i \in \{ 1,\ldots,n_B\}$. The difficulty is that the discrepancy that
we allow for row $i$ changes dynamically as the number of active variables decreases. 
The sequence of iterations  $T_r := \{ t \mid 2^{r} \leq |J_{t-1}| < 2^{r+1}\}$, in which the number of active variables is between $2^r$ and $2^{r+1}$ is termed 
\emph{phase $r$}. Let $\delta(r) := G^{-1}\left(\frac{\mu_i}{10}2^{r}\right)\cdot 2^{(r+1)/2}$ be an upper bound on the discrepancy 
bound that is imposed to row $i$ during this phase. 
Again we can write
\[
  |B_i\chi^{\ell}| = \Big| \sum_{t=1}^{\ell} B_i\gamma_t \Big| = \sum_{t=1}^{\ell} \Big|g_t^T s\Big( \sum_{j=1}^m B_{ij} v_j^t \Big)\Big|
= \sum_{r\geq 0} \overbrace{\Big|s \sum_{t \in T_r} g_t^T \underbrace{\left(\sum_{j=1}^m B_{ij}v_j^t \right)}_{=: u^t}\Big|}^{=:X(r)}
\]
and $\|u^t\|_2 \leq \delta(r)$. By $X(r)$ we denote the discrepancy that we suffer in phase $r$. 
We somewhat expect that $E[X(r)] = O(\delta(r))$.
In fact, this is true even with a strong tail bound.
\begin{lemma} \label{lem:Concentration-result-phase-r}
For all $r \geq 0$ and $\lambda \geq 0$, $\Pr\left[X(r) > \lambda \cdot \delta(r)\right] \leq 3\cdot 2^{-\lambda/6}$.
\end{lemma}
\begin{proof}
If we suffer a large discrepancy in a single phase, then either the phase lasted
much longer than $O(\frac{1}{s^2})$ iterations or the phase was short but the discrepancy 
exceeded the standard deviation by a large factor. However both is unlikely.
More formally, for the event $X(r) > \lambda\cdot \delta(r)$ to happen, at least one of the following events
must occur
\begin{itemize}
\item $E_1$: $|T_r| \geq \frac{\lambda}{4}\cdot \frac{16}{s^2}$
\item $E_2$: Within the first $\frac{\lambda}{4} \cdot \frac{16}{s^2}$ iterations of phase $r$, a discrepancy of $\lambda\cdot \delta(r)$ is reached.
\end{itemize}
By Lemma~\ref{lem:NumVariablesHalfAfter16-over-s2It}, one has $\Pr[E_1] \leq (\frac{1}{2})^{\lfloor \lambda/4 \rfloor}$. Furthermore
\[
  \Pr[E_2] \leq \Pr\Big[\Big|\sum_{\textrm{first }\frac{4\lambda}{s^2}\textrm{ it.}  t \in T_r} s\cdot g_t^Tu^t \Big| > \underbrace{\frac{\sqrt{\lambda}}{2}\cdot s\delta(r)\cdot \sqrt{\frac{4\lambda}{s^2}}}_{=\lambda\cdot \delta(r)}\Big] \leq 2e^{-\lambda/8}
\]
by again applying Lemma~\ref{lem:MartingaleBound} with parameters $\lambda' := \frac{\sqrt{\lambda}}{2}; \; \delta' := s\cdot \delta(r); \; n' := \frac{4\lambda}{s^2}$.
The claim follows since
$\Pr[X(r) > \lambda\cdot \delta(r)] \leq \Pr[E_1] + \Pr[E_2] \leq 2^{-\lfloor \lambda/4 \rfloor } + 2e^{-\lambda/8} \leq 4\cdot 2^{-\lambda/6}$.
\end{proof}
By the union bound, we could easily bound the probability that any phase $r$ has $X(r) > \lambda\cdot \delta(r)$ by
$4\cdot 2^{-\lambda/6} \cdot \log m$ and thus $\Pr[|B_i\chi^{\ell}| > \lambda \cdot C/\sqrt{\mu_i}] \leq 4\cdot 2^{-\lambda/6} \cdot \log m$. However, we 
can avoid the $\log m$ term by observing that the bound on $|B_ix - B_iy|$  in the proof of Theorem~\ref{thm:GeneralRoundingTheorem}
receives the largest contributions within the small window, when the number of 
active variables is $\Theta(1/\mu_i)$.
Outside of this window, we have a lot of slack, that we can use here. 

We call a phase $r$ \emph{bad}, if $|X(r)| > \delta(r)\cdot \big(\lambda + |r - \log(\frac{60}{\mu_i})|\big)$ (and \emph{good} otherwise). 
Note that the term $|r - \log (\frac{60}{\mu_i})|$ is indeed minimized
if $2^r = \Theta(\frac{1}{\mu_i})$. Then we can upper bound the probability that any phase is bad by
\[
\sum_{r\geq 0} \Pr\Big[X(r) > \delta(r)\cdot \left(\lambda+\left|r - \log\left(\frac{60}{\mu_i}\right)\right|\right)\Big] \stackrel{\textrm{Lem.~\ref{lem:Concentration-result-phase-r}}}{\leq} 2\sum_{z\geq 0} 4\cdot 2^{-(\lambda+z)/6} \leq 80\cdot 2^{-\lambda/6} 
\]
It remains to prove that if all phases $r$ are good, then $|B_i\chi| \leq \lambda \cdot O(\sqrt{1/\mu_i})$.
Hence we consider
\begin{eqnarray*}
 |B_i\chi| &\stackrel{(*)}{\leq}& \sum_{r\geq0} G^{-1}\Big(\frac{\mu_i2^r}{10}\Big)\cdot \sqrt{2^{r+1}} \cdot \left(\lambda+\left|r - \log\left(\frac{60}{\mu_i}\right)\right|\right) \\
&\stackrel{(**)}{\leq}& \sqrt{\frac{120}{\mu_i}} \sum_{z\in \setZ} G^{-1}(6\cdot 2^{z})\cdot 2^{z/2}(\lambda+|z|) \\
&\stackrel{(***)}{\leq}& \lambda \cdot C\sqrt{1/\mu_i})
\end{eqnarray*}
for some constant $C>0$.
In $(*)$ we assumed that all phases were good and in $(**)$ we substitute $z := r - \log_2( \frac{60}{\mu_i})$ (or equivalently $2^z = 2^r \cdot \frac{\mu_i}{60}$).
Eventually we recall that already in the proof of Theorem~\ref{thm:GeneralRoundingTheorem} we saw that the 
series  $\sum_{z\in \setZ} G^{-1}(6\cdot 2^{z})\cdot 2^{z/2}$ converges geometrically, which is not affected by adding a polynomial term like $|z|$, hence giving $(***)$
(here we also use $\lambda \geq 1$). This \emph{almost}
concludes the proof of Theorem~\ref{thm:ComputabilityOfLowDiscrepancyColoring}, since with probability
at most $\frac{1}{m^2} + \frac{2}{n^4m^4}\cdot \ell \leq \frac{1}{2}$ (for $m$ large enough)
the algorithm produces a \emph{failure}, i.e. not all variables are frozen after $O(\frac{1}{s^2}\log m)$  iterations. 
In this case, we simply repeat the algorithm until it was successful. 
Then the actual tail bound that we obtain is a conditional probability
\begin{equation}  \label{eq:ConditionalProbForLowDisc}
  \Pr[ |B_i\chi| > \lambda C/\sqrt{\mu_i} \mid \textrm{run successful}] \leq \frac{\Pr[ |B_i\chi| > \lambda C/\sqrt{\mu_i}]}{\Pr[\textrm{run successful}]} \leq 2\cdot 80\cdot 2^{-\lambda/6}
\end{equation}
(same for $|A_i\chi|$) and the expected running time is polynomial.

\subsection{The Constructive Rounding Theorem} 

Now that we can compute efficiently full colorings $\chi$ such that $A\chi,B\chi \approx \mathbf{0}$, 
it is not difficult anymore to give an algorithmic version of our main theorem. 
\begin{theorem} \label{thm:GeneralRoundingTheoremConstructiveVersion}
Let $A \in \setQ^{n_A × m}$,  $B \in [-1,1]^{n_B × m}$,  $\Delta = (\Delta_1,\ldots,\Delta_{n_A}) > \mathbf{0}$, $\mu_1,\ldots,\mu_{n_B} > 0$ ($\sum_{i=1}^{n_B} \mu_i \leq 1$, $m \geq 2$)
and $c \in [-1,1]^{m}$ be given as input, such that $\forall J \subseteq \{ 1,\ldots,m\}: H_{\Delta}(A^J) \leq \frac{|J|}{10}$.
Then there is a constant $C'>0$ and a randomized algorithm with expected polynomial running time 
which obtains a $y\in\{0,1\}^{m}$ such that
\begin{itemize}
\item Preserved expectation:  $E[c^Ty] = c^Tx$, $E[Ay]=Ax$, $E[By] = Bx$.
\item Bounded difference: $|c^Tx - c^Ty| \leq C'$;
$|A_ix - A_iy| \leq C'\sqrt{\log n} \cdot \sqrt{\log \min\{ n,m\}}\cdot \Delta_i$ for all $i=1,\ldots,n_A$ ($n := n_A + n_B$);
$|B_ix - B_iy| \leq C'\log(\frac{2}{\mu_i}) \sqrt{1/\,\mu_i}$ for all $i=1,\ldots,n_B$.
\item Tail bounds: For all $\lambda \geq 0$ and all $i$:
  \begin{itemize}
  \item $\Pr[|A_ix - A_iy| > \lambda \cdot C'\sqrt{\log \min\{ n,m\}}\cdot \Delta_i] \leq 2\cdot 2^{-\lambda^2}$ 
  \item $\Pr[|B_i x - B_i y| > \lambda \cdot C'/\sqrt{\mu_i}] \leq 2\cdot 2^{-\lambda}.$ 
  \end{itemize}
\end{itemize}
\end{theorem}

\begin{proof}
Again we can append $c$ as an additional row to $B$, hence we ignore the objective function from now on.
As described in the proof of Theorem~\ref{thm:GeneralRoundingTheorem},
we can assume that $m\leq n$ and that entries of $x$ have a finite dyadic expansion with $K$ bits. 
We perform the following algorithm:
\begin{enumerate}
\item[(1)] FOR $k := K,K-1,\ldots,1$ DO
  \begin{enumerate}
  \item[(2)] $J := \{ j \in \{1,\ldots,m\} \mid x_j\textrm{'s }k\textrm{th bit is }1\} $
  \item[(3)] Repeat computing $\chi^{(k)} : J \to \{ ± 1 \}$ according to Theorem~\ref{thm:ComputabilityOfLowDiscrepancyColoring}
until $\chi^{(k)}$ is \emph{good}, i.e. until
   \begin{itemize}
     \item $|A_i\chi^{(k)}| \leq C'\cdot \sqrt{\log n}\cdot \sqrt{\log m} \cdot \Delta_i$ for all $i=1,\ldots,n_A$
     \item $|B_i\chi^{(k)}| \leq C' \log(\frac{2}{\mu_i})/\sqrt{\mu_i}$ for all $i=1,\ldots,n_B$
   \end{itemize}
  \item[(4)] Update $x := x + (\frac{1}{2})^k\chi^{(k)}$
  \end{enumerate}
\end{enumerate}
Let $y$ be the integral vector obtained at the end.
For $C'$ large enough, by Theorem~\ref{thm:ComputabilityOfLowDiscrepancyColoring} each run to compute coloring $\chi^{(k)}$ has 
$\Pr[ |A_i\chi| > C' \sqrt{\log n} \cdot \sqrt{\log m} \cdot \Delta_i] \leq \frac{1}{4n}$ and $\Pr[|B_i\chi| > C' \log(\frac{2}{\mu_i}) / \sqrt{\mu_i}] \leq  \frac{1}{4}\mu_i$.
By the union bound, each run of (3) is good with probability at least $\frac{1}{2}$. 
By Equation~\eqref{eq:ConditionalProbForLowDisc} concerning conditional probabilities, this only worsens
the tail bounds provided by Theorem~\ref{thm:ComputabilityOfLowDiscrepancyColoring} for $B_i\chi^{(k)}$ and $A_i\chi^{(k)}$  by a factor of at most $2$. In any case we have a guarantee that
\[
  |A_ix-A_iy| \leq \sum_{k \geq 1} \left(\frac{1}{2}\right)^k\big|A_i\chi^{(k)}\big|
\leq \sum_{k \geq 1} \left(\frac{1}{2}\right)^{k}\cdot C'\sqrt{\log n \cdot \log m} \cdot \Delta_i = C'\sqrt{\log n \cdot \log m} \cdot \Delta_i.
\]
and analogously $|B_ix - B_iy| \leq C' \log(\frac{2}{\mu_i})/\sqrt{\mu_i}$.
The algorithm behind Theorem~\ref{thm:ComputabilityOfLowDiscrepancyColoring} is fully symmetric, i.e. $E[A\chi^{(k)}] = \mathbf{0}$ for all $k$, hence $E[Ay]=Ax$ 
(and similar $E[By] = Bx$). It remains to prove the tail bounds. We may assume that $\lambda \geq 1$, 
since otherwise, the desired probabilities $2\cdot 2^{-\lambda}$ and $2\cdot 2^{-\lambda^2}$ exceed $1$ anyway.


%
Note that if $|B_i\chi^{(k)}| \leq \frac{(\lambda+k)C'}{4\sqrt{\mu_i}}$ holds for all $k$, then %
\[
|B_ix - B_iy| = \Big| \sum_{k\geq 1} \left(\frac{1}{2}\right)^{k} B_i\chi^{(k)} \Big| 
\leq \sum_{k\geq 1} \left(\frac{1}{2}\right)^{k} \frac{C'\cdot (\lambda+k)}{4\sqrt{\mu_i}} \leq \lambda \frac{C'}{\sqrt{\mu_i}}
\]
since $\sum_{k\geq 1} (\frac{1}{2})^{k}\leq 2$, $\sum_{k\geq 1} (\frac{1}{2})^k k \leq 2$ and $\lambda \geq 1$.
Hence we can use the bound
\[
  \Pr\left[|B_ix - B_iy| > \lambda \cdot \frac{C'}{\sqrt{\mu_i}}\right] \leq \sum_{k\geq 1} \Pr\left[|B_i\chi^{(k)}| > \frac{(\lambda+k)C'}{4C} \frac{C}{\sqrt{\mu_i}}\right] \stackrel{\textrm{Thm.~\ref{thm:ComputabilityOfLowDiscrepancyColoring}}}{\leq} 
\sum_{k\geq 1} 320\cdot 2^{-\frac{(\lambda+k)C'}{24C}}  \leq 2\cdot 2^{-\lambda}
\]
for $C'$ large enough. Similarly
\[
  \Pr\left[|A_ix - A_iy| > \lambda \cdot C'\sqrt{\log m}\cdot \Delta_i\right]
\leq \sum_{k\geq 1} \Pr\Big[|A_i\chi^{(k)}| > \frac{(\lambda+k)C'}{4C}C\sqrt{\log m}\Delta_i\Big]
\stackrel{\textrm{Thm.~\ref{thm:ComputabilityOfLowDiscrepancyColoring}}}{\leq} \sum_{k\geq 1} 8e^{-\frac{(\lambda+k)^2 C'^2}{32C}} \leq 2\cdot 2^{-\lambda^2}
\]
again for $C'$ large enough and $\lambda \geq 1$.
\end{proof}

\section{How to solve the LP relaxations\label{Appendix:SolvingLPrelaxation}}

All linear programs for which we provided rounding procedures were of the form
$\min\{ c^Tx \mid \sum_{S \in \setS} x_S\mathbf{1}_S = \mathbf{1}, x \geq \mathbf{0} \}$, i.e. they all have an exponential number of variables.
So, we should explain how such programs can be solved. In fact, the first polynomial time
algorithm was proposed by \cite{KarmarkarKarp82} in the case of {\sc Bin Packing}. Their approach 
solves the dual $\max\{ \sum_{i=1}^n y_i \mid \sum_{i\in S} y_i \leq c_S \; \forall S \in \setS \}$ up to an arbitrarily small additive error
using the 
Gr{ö}tschel-Lovász-Schrijver variant of the Ellipsoid method~\cite{GLS-algorithm-Journal81}.
The error term cannot be avoided, since a {\textsc{Partition}} instance could be decided by inspecting
whether $OPT_f\leq2$ or not.
The only additional prerequisite for the Karmarkar-Karp algorithm is an FPTAS for the dual separation problem (i.e. 
given dual prices $y_1,\ldots,y_n \geq 0$, find a $(1-\varepsilon)$-approximation to $\max\{ \frac{1}{c_S}\sum_{i\in S} y_i \mid S \in \setS\}$).
Note that the same result is implied by the framework of Plotkin, Shmoys and 
Tardos~\cite{FractionalPackingAndCovering-PlotkinShmoysTardos-Journal95} without using general LP solvers.

It follows implicitly from both papers~\cite{KarmarkarKarp82,FractionalPackingAndCovering-PlotkinShmoysTardos-Journal95} that for any set family $\setS \subseteq 2^{[n]}$ that admits an FPTAS for the dual separation 
problem, the corresponding column-based LP can be solved within an arbitrarily small additive error. 

However, we are not aware of an explicit proof of this fact in the literature. 
Hence, to be self-contained we provide all the details here. 
Our focus lies on giving a short and painless analysis, rather than giving the
best bounds on the running time.
Our starting point is the following theorem from \cite{FractionalPackingAndCovering-PlotkinShmoysTardos-Journal95}
(paraphrased to make it self-contained).
\begin{theorem}[Plotkin, Shmoys, Tardos~\cite{FractionalPackingAndCovering-PlotkinShmoysTardos-Journal95}]  \label{thm:PlotkinShmoysTardosCoveringAlgo}
Let $A \in \setR_{\geq0}^{n×m}$ be a matrix and $P \subseteq \setR_{\geq0}^m$ be a convex set. 
Given $0<\varepsilon<1$, $n\in\setN$, $b \in \setQ_{>0}^n$, $\rho \geq \max_{i=1,\ldots,n} \max_{x\in P} \frac{A_ix}{b_i}$ as input. 
Then there exists an algorithm $\textsc{Cover}$ which either computes an $x\in P$ with $Ax \geq (1-\varepsilon)b$ or asserts that there
is no $x\in P$ with $Ax \geq b$. This algorithm calls $K:=O(n + \rho\log^2(n) + \frac{\rho}{\varepsilon^2}\log(\frac{n}{\varepsilon}))$ many times the following
oracle (with $\varepsilon' \in \{ \varepsilon/2, 1/6\}$)
\begin{quote}
$\textsc{Subroutine:}$ Given $0<\varepsilon'<1, y \in \setR_+^n$ as input. Find a $\tilde{x} \in P$ such that $y^TA\tilde{x} \geq (1-\varepsilon')\max\{ y^TAx \mid x \in P\}$.
\end{quote}
Assuming that for any $\tilde{x}$, the vector $A\tilde{x}$ can be evaluated in time $O(n)$, 
the additional running time of  $\textsc{Cover}$ is $O(K \cdot n)$.
\end{theorem}
On an intuitive level, the algorithm of \cite{FractionalPackingAndCovering-PlotkinShmoysTardos-Journal95} 
maintains at any iteration some vector $x\in P$. Then for every element $i\in[n]$ 
one defines certain \emph{dual prizes} $y_i$ which are decreasing in $\frac{A_ix}{b_i}$. 
In other words, uncovered elements will
receive a high dual price $y_i$; covered ones receive a low price. Then one computes a vector $\tilde{x}$ which
(approximately) maximizes the dual prices, meaning that $\tilde{x}$ has a large incentive to cover elements $i$ with $A_ix \ll b_i$.
Then one replaces $x$ by a convex combination of $x$ and $\tilde{x}$ and iterates.

In the following we show how Theorem~\ref{thm:PlotkinShmoysTardosCoveringAlgo} can be used to solve \eqref{eq:GeneralLP}.
\begin{theorem} \label{thm:PolytimeSolvabilityOfLPs}
Let $\setS \subseteq 2^{[n]}$ be a family of sets with cost function $c : \setS \to ]0,1]$ (assume $c(S)$ can be evaluated in time $O(n)$)
such that for any $y \in \setQ_+^{n}$ given as input, one can find an $S^* \in \setS$ with $\sum_{i\in S^*} y_i \geq (1-\varepsilon)\cdot\max\{ \frac{1}{c(S)} \sum_{i\in S} y_i \mid S \in \setS \}$
in time $T(n,\varepsilon)$.  
Then for any given $n/2 \geq \delta > 0$, one can find a basic solution $x$ of the LP
\[
  OPT_f = \min\Big\{ c^Tx \mid  \sum_{S \in \setS} x_S\mathbf{1}_S \geq \mathbf{1}, \; x \geq \mathbf{0} \Big\}
\]
of cost $c^Tx \leq OPT_f + \delta$ in time $O\left(\frac{n^4}{\delta^2}\ln(\frac{n}{\delta})\right)\cdot(\frac{1}{\delta}T(n,\Omega(\delta/n)) + n^2)$. 
\end{theorem}
\begin{proof}
Since no set costs more than $1$, one has $OPT_f \leq n$. 
By trying out $O(n/\delta)$ values, we may assume to know a value $r$ with $OPT_f \leq r \leq OPT_f + \frac{\delta}{2}$.
We define $P = \{ x \in \setR_{\geq0}^{\setS} \mid \sum_{S \in \setS} c_Sx_S = r\}$ and a matrix\footnote{$A$ and $P$ are \emph{defined}, but
not explicitly \emph{computed}.} $A \in \{0,1\}^{n × \setS}$ by
\[
 A_{iS} = \begin{cases} 1 & i \in S \\ 0 & \textrm{otherwise} \end{cases}
\]
as well as $b=(1,\ldots,1)$. We choose 
\[
  \rho := n \geq r \geq \max_{i=1,\ldots,m} \max_{x\in P} \frac{A_ix}{b_i}
\]
and $\varepsilon := \frac{\delta}{4n}$. The next step is to design the $\textsc{Subroutine}$. 
Hence, let a vector of dual prices $y\in\setQ_{\geq0}^n$ and a parameter $\varepsilon' > 0$ be given as input.
Then we compute a set $S^* \in \setS$
with $\sum_{i\in S^*} y_i \geq (1-\varepsilon')\cdot\max\{ \frac{1}{c_S}\sum_{i\in S} y_i \mid S \in \setS \}$ in time $T(n,\varepsilon')$. 
Observe that the vertices of $P$ are of the form $\frac{r}{c_S}e_S$ and
\[
y^TA\left(\frac{r}{c_S}e_S\right) = \frac{r}{c_S}\sum_{i=1}^n y_iA_{iS} = \frac{r}{c_S}\sum_{i\in S} y_i
\]
hence, the vector $\tilde{x} := \frac{r}{c_S}\cdot e_{S^*}$ is the desired $(1-\varepsilon')$-approximation for $\max\{ y^TAx \mid x \in P\}$.

Applying Theorem~\ref{thm:PlotkinShmoysTardosCoveringAlgo} yields a vector $x \in P$ with $Ax \geq (1-\varepsilon)\mathbf{1}$. 
Hence the slightly scaled vector $x'=\frac{x}{1-\varepsilon}$ is feasible and has cost $\frac{r}{1-\varepsilon} \leq (OPT_f + \frac{\delta}{2})\cdot(1+2\varepsilon) \leq OPT_f + \delta$.
To turn $x'$ into a basic solution, we consider $x''$ with $x''_S = x_S' > 0$ for precisely $n+1$ many sets (and $x''_S=0$ otherwise). Then by Gauss elimination we find a $x''' \geq \mathbf{0}$ with $Ax'''=Ax''$, $|\textrm{supp}(x''')|\leq n$ and $c^Tx''' \leq c^Tx''$ in time $O(n^3)$ and replace the corresponding values in $x'$ by those in $x'''$. After iteration this at most $K$ times, 
we obtain the desired basic solution. The total running time is
\[
  \frac{n}{\delta}\cdot K \cdot T(n,\Omega(\varepsilon)) + O(K\cdot n^3) \leq O\left(\frac{n^4}{\delta^2}\ln(\frac{n}{\delta})\right)\cdot(\frac{1}{\delta}T(n,\Omega(\delta/n)) + n^2) 
\] 
since $\textsc{Subroutine}$ was called at most $K = O(\frac{n^3}{\delta^2}\ln(\frac{n}{\delta}))$ times.
\end{proof}
Here the running times are w.r.t. the RAM model, were any arithmetic operation accounts with unit cost.

\subsection*{Applications for considered problems}

In order to solve the considered LPs up to any additive error term, it 
suffices to provide an FPTAS for each of the corresponding dual separation problems.
\begin{itemize}
\item $\textsc{Bin Packing}:$ The dual separation problem is $\max\{ \sum_{i\in S} y_i \mid \sum_{i\in S} s_i \leq 1 \}$ which is known as $\textsc{Knapsack}$
problem
and admits an FPTAS in time $T(n,\varepsilon) = O(n/\varepsilon^2)$ (see e.g.~\cite{ApproximationAlgorithmsVazirani01}).
\item $\textsc{Bin Packing With Rejection}$: 
We compute a $(1-\varepsilon)$-approximation $S^*$ to $\max\{ \sum_{i\in S} y_i \mid \sum_{i=1}^n s_i \leq 1\}$ in time $O(n/\varepsilon^2)$ and
compare it to the values $\frac{y_i}{\pi_i}$ for $i=1,\ldots,n$ and output either $S^*$ or some $\{ i\}$, whoever yields
the largest value, hence again $T(n,\varepsilon) = O(n/\varepsilon^2)$.
\item $\textsc{Train Delivery}$. For any  $k\in\{1,\ldots,n\}$, let $S^k$ be a $(1-\varepsilon)$-approximate
solution to $\max\{ \sum_{i\in S} y_i \mid \sum_{i\in S} s_i \leq 1, S \subseteq \{ i \mid p_i \leq p_k\} \}$. We output the set $S^k$ maximizing $\sum_{i \in S^k} \frac{y_i}{p^k}$.
This can be done in time $T(n,\varepsilon) = O(n^2/\varepsilon^2)$.
\end{itemize}

\section{Omitted proof for Lemma~\ref{lem:EntropyBound}\label{sec:OmittedProofForEntropyBound}}

\begin{lemma*}[Lemma \ref{lem:EntropyBound}]
Let $\alpha \in \setR^{m}$ be a vector and $\Delta > 0$. For  $\lambda = \frac{\Delta}{\|\alpha\|_2}$,
\[
 \mathop{H}_{\chi \in \{ ± 1\}^m}\left(\left\lceil \frac{ \alpha^T\chi}{2\Delta} \right\rfloor\right) \leq
G(\lambda) := \begin{cases} 9e^{-\lambda^2/5} & \textrm{if } \lambda \geq 2 \\
\log_2( 32 + 64/\lambda) & \textrm{if } \lambda < 2 
\end{cases} 
\]
\end{lemma*}
\begin{proof}
We distinguish 2 cases.
\emph{Case $\lambda \geq 2$.} Let $p_i := \Pr[Z = i]$. Note that $X := \alpha^T\chi = \sum_{j=1}^m \alpha_j\cdot\chi_j$ 
is the sum of independently distributed random variables 
$\chi_j\cdot \alpha_j = ± |\alpha_j|$ with mean $0$. 
For $i\geq1$, 
\[
 p_i \leq \Pr[X \geq (2i-1)\Delta] \stackrel{\textrm{Lem.~\ref{lem:AzumaHoeffdingInequality}}}{\leq} 
 2e^{- \frac{(2i-1)^2 \Delta^2}{2\|\alpha\|_2^2}}
= 2e^{ - (2i-1)^2\lambda^2/2}
\stackrel{\lambda\geq2}{\leq} e^{ - (2i-1)^2\lambda^2/4}
\]
The entropy, stemming from  $i\geq1$ is fairly small, namely
\begin{eqnarray*}
\sum_{i\geq1} p_i \log_2\left(\frac{1}{p_i}\right)
\leq \sum_{i\geq1} e^{ - (2i-1)^2\lambda^2/4}\cdot \log_2\left(\frac{1}{e^{ - (2i-1)^2\lambda^2/4}}\right) 
= \sum_{i\geq1} e^{ - (2i-1)^2\lambda^2/4} \cdot \frac{(2i-1)^2 \lambda^2}{4\cdot\ln(2)} 
\leq 3\cdot e^{-\lambda^2/5}
\end{eqnarray*}
Here we use that $p_i\leq \frac{1}{e}$ and $x\cdot\log_2(\frac{1}{x})$ is monotone increasing for $x\in[0,\frac{1}{e}]$.
Furthermore $p_0 \geq 1 - \Pr[|X| \geq \Delta] \geq 1-2e^{-\lambda^2/2}$, 
hence the event $\{Z=0\}$ is so likely that it also does not contribute much entropy.
\[
  p_0 \log \left(\frac{1}{p_0}\right) \leq 2\cdot (1 - p_0) \leq 4e^{-\lambda^2/2} \leq 2e^{-\lambda^2/5}
\]
Adding up also the entropy for $i<0$, we obtain $H(Z) \leq 8\cdot e^{-\lambda^2/5}$.

\emph{Case $\lambda < 2$.} Define $L := \lceil\frac{2}{\lambda}\rceil > 1$ and $\Delta' := \Delta \cdot L$. Then we can 
express $Z =  L\cdot \big\lceil\frac{X}{2L\Delta}\big\rfloor + Z''$
such that $Z''$ attains just $L$ different values (and hence $H(Z'') \leq \log_2(L)$).
Let $\lambda' := \Delta'/ \|\alpha\|_2 \geq 2$, then

\[
H(Z) \stackrel{(*)}{\leq} H\Big( \Big\lceil\frac{X}{2\Delta'}\Big\rfloor\Big) + H(Z'') \stackrel{(**)}{\leq} 9e^{-{\lambda'}^2/5} + \log_2(L) 
\leq \underbrace{9\cdot e^{-4/5}}_{<5} + \log_2 \left(\frac{2}{\lambda}+1\right) 
\]
In $(*)$, we use the subadditivity of the entropy function. In $(**)$
we use that $H\big( \big\lceil\frac{X}{2\Delta'}\big\rfloor\big)$ can be bounded by case (1).
\end{proof}

\end{document}